\documentclass[12pt,journal,draftcls,letterpaper,onecolumn]{IEEEtran}
\usepackage{amsmath,amssymb,times}
\usepackage{multirow,epsfig,cite,psfrag}
\usepackage{times, amsmath}
\usepackage{stfloats}
\usepackage{dsfont}
\usepackage{breqn}
\usepackage{color}
\usepackage{enumerate}
\usepackage{amsfonts}
\usepackage{amsthm}
\usepackage{amsmath}
\usepackage{array}
\usepackage{bbm}
\usepackage{cite}
\usepackage{dsfont}
\usepackage{epsfig}
\usepackage{float}
\usepackage{algorithm}
\usepackage[T1]{fontenc}
\usepackage{graphicx}
\usepackage{indentfirst}
\usepackage{multicol}
\usepackage{subfigure}
\usepackage{url}
\usepackage{multirow}
\usepackage{algpseudocode}
\usepackage{color}
\usepackage{epstopdf}
\usepackage{algpseudocode}
\usepackage[normalem]{ulem}
\usepackage{float}
\usepackage[nolist]{acronym}
\usepackage{booktabs}
\usepackage{subfigure}
\usepackage{graphicx}

\usepackage{graphicx}
\usepackage{caption}
\usepackage{algorithm}
\usepackage{algpseudocode}

\newtheorem{theorem}{\bf Theorem}

\newtheorem{lemma}{\bf Lemma}

\begin{document}
\title{\LARGE Sleeping Multi-Armed Bandit Learning for Fast Uplink Grant Allocation in Machine Type Communications}
	\author{\IEEEauthorblockN{Samad Ali, \emph{Student Member, IEEE}, Aidin Ferdowsi, \emph{Student Member, IEEE}, Walid Saad, \emph{Senior Member, IEEE}, Nandana Rajatheva, \emph{Senior Member, IEEE}, and Jussi Haapola, \emph{Member, IEEE}}
		\thanks{A preliminary version of this work appeared in the IEEE  GLOBECOM 2018 Workshops \cite{samad_globecom}.}
		\thanks{S. Ali, N. Rajatheva and J. Haapola are with the Centre for Wireless Communications (CWC), University of Oulu, Finland. Emails: \{samad.ali, nandana.rajatheva, jussi.haapola\}@oulu.fi. A. Ferdowsi and W. Saad are with Wireless@VT, Bradley Department of Electrical and Computer Engineering, Virginia Tech, Blacksburg, VA, USA, Emails: \{aidin, walids\}@vt.edu.}}
\maketitle
\vspace{-2cm}
\begin{abstract}
Scheduling fast uplink grant transmissions for machine type communications (MTCs) is one of the main challenges of future wireless systems. In this paper, a novel fast uplink grant scheduling method based on the theory of multi-armed bandits (MABs) is proposed. First, a single quality-of-service metric is defined as a combination of the value of data packets, maximum tolerable access delay, and data rate. Since full knowledge of these metrics for all machine type devices (MTDs) cannot be known in advance at the base station (BS) and the set of active MTDs changes over time, the problem is modeled as a sleeping MAB with stochastic availability and a stochastic reward function. In particular, given that, at each time step, the knowledge on the set of active MTDs is probabilistic, a novel probabilistic sleeping MAB algorithm is proposed to maximize the defined metric. Analysis of the regret is presented and the effect of the prediction error of the source traffic prediction algorithm on the performance of the proposed sleeping MAB algorithm is investigated. Moreover, to enable fast uplink allocation for multiple MTDs at each time, a novel method is proposed based on the concept of best arms ordering in the MAB setting. Simulation results show that the proposed framework yields a three-fold reduction in latency compared to a random scheduling policy since it prioritizes the scheduling of MTDs that have stricter latency requirements. Moreover, by properly balancing the exploration versus exploitation tradeoff, the proposed algorithm can provide system fairness by allowing the most important MTDs to be scheduled more often while also allowing the less important MTDs to be selected enough times to ensure the accuracy of estimation of their importance.
\end{abstract}
\begin{IEEEkeywords} Machine Type Communications, Scheduling, Fast Uplink Grant, Multi-armed Bandits, Internet of Things \end{IEEEkeywords}
\section{Introduction} \label{sec:introduction}
The fifth generation (5G) of cellular communication networks is expected to support Internet of Things (IoT) \cite{IoTin5G} services and applications such as virtual reality \cite{mingzheVR}, autonomous vehicles \cite{aidin_its_magazine}, and unmanned areal vehicles \cite{mozaffariUAV}. To enable such emerging IoT applications, 5G systems must have native support for machine type communications (MTCs). In contrast to enhanced mobile broadband (eMBB) services that require high data rates for large data packets, in MTC, a large number of machine-type-devices (MTDs) must communicate small data packets \cite{dawyM2MMagazine}. Due to the heterogeneous nature of IoT applications, MTC data packets have fundamentally novel requirements in terms of latency, reliability, and security \cite{aidin-deeplearning-journal, aidin_ICC}. These requirements bring forward new cellular networking challenges that include random access channel congestion, signaling overhead management, and a need for satisfying various quality-of-service (QoS) requirements for different IoT applications \cite{schulz2017latency}. Moreover, optimizing the wireless system in terms of throughput and spectral efficiency is challenging since it is hard to acquire uplink channel state information (CSI) of transmitting MTDs at the base station (BS). Therefore, reducing the signaling overhead and latency, while avoiding random access channel congestion are important open problems in MTCs.

MTC can be categorized into two groups depending on whether scheduling requests are sent by MTDs or not. The first MTC group is \emph{coordinated transmission}, in which MTDs perform a random access process and the BS schedules MTDs, similar to conventional cellular systems. Clearly, this method is inefficient since the data packets are small and, hence, the signaling to data packet size ratio is large. In the second method, known as \emph{uncoordinated transmission}, to reduce the signaling overhead, MTDs choose a random uplink radio resource and transmit their data without sending any scheduling request. Both approaches can suffer from severe collisions among transmissions due to the fact that the number of MTDs is often much larger than the number of available resources. In a coordinated transmission, collisions can occur during random access while in the uncoordinated method, they occur during packet transmission. In a massive MTC \cite{MassiveM2M} scenario, such problems become even more challenging to address. The authors in \cite{surveyofaccess} and \cite{RACHM2M2} provide an extensive overview of several proposed solutions for such problems. One possible solution is known as access class barring (ACB) \cite{otpimalACB} where different access classes are assigned to MTDs and in massive access scenario, MTDs with lower class are barred from a transmission. The authors in \cite{nora} leverage ideas from non-orthogonal multiple access and apply them to the random access process so as to identify random access requests from multiple MTDs with the same preamble. In \cite{RACH_BloomFiltering} the authors propose a method to reduce the signaling overhead of random access using signatures and Bloom filtering. Correlation between transmission patterns of different MTDs is exploited in \cite{rach_correlated} to optimize the random access process by reducing the collisions. To avoid wasting radio resources in random access collisions, the authors in \cite{newRACH} propose to attach the MTD identity information in the physical random access channel which will prevent the BS from allocating uplink resources to devices that collided. Clearly, most of the research in this area \cite{otpimalACB, nora, RACH_BloomFiltering, rach_correlated, newRACH} is focused on optimizing random access process for MTC and solving problems associated with collisions. 

For uncoordinated transmission, in \cite{madueno2014reliable}, the authors present a resource allocation approach for a massive number of devices with reliability and latency guarantees. Meanwhile, the work in \cite{abuzainab2016cognitive} presents a game-theoretic model for optimizing the coexistence of MTDs with cellular users in the uplink period. Throughput and outage analysis of uncoordinated non-orthogonal multiple access (NOMA) for massive MTC are presented in \cite{grantfreemassiveNOMA} in which the authors compare the performance of successive joint decoding (SJD) to successive interference cancellation (SIC). A dynamic Compressed sensing (CS) multi-user detection is used in \cite{CS_GrantFree_NOMA} to exploit the user correlation in uncoordinated uplink NOMA for joint user activity detection and decoding, which has much better performance compared to conventional CS based schemes. Even though these prior solutions can improve the performance of MTCs, coordinated access still suffer from heavy signaling overheard and collisions \cite{MassiveM2M, surveyofaccess, otpimalACB, RACH_BloomFiltering, rach_correlated, newRACH}. Moreover, uncoordinated transmissions still also experience non-negligible collisions, particularly in massive access scenarios \cite{madueno2014reliable, abuzainab2016cognitive, grantfreemassiveNOMA, CS_GrantFree_NOMA}. The main drawback of this prior art is that it relies solely on random access process (for sending scheduling requests in coordinated transmission or sending data packets in the uncoordinated scheme) whose performance is optimal only when the number of competing devices is equal to the number of available resources. This clearly does not hold in massive MTC cases since the number of radio resources is limited, and hence, novel solutions are needed to address the uplink resource allocation problem for massive MTCs.

To address the challenges of random access congestion, collisions and high signaling overhead, a middle ground, from the point of view of the uplink resource allocation method between a) fully scheduled by using scheduling requests and b) uncoordinated transmission, can be achieved by using the concept of a \emph{fast uplink grant} \cite{3GPP-fastuplinkgrant} and \cite{LTE14Outlook}. In the fast uplink grant, a BS sends an uplink grant to MTDs without MTDs sending scheduling requests. If the MTDs have data to transmit, they proceed with the transmission, otherwise, the radio resource is wasted \cite{LTE14Outlook}. An overview of challenges and opportunities of the fast uplink grant is provided in \cite{Samad-fastuplinkgrant} and two main challenges associated with the fast uplink grant are outlined. The first challenge is that the set of the MTDs that have data to transmit should be known to the BS. The second challenge is the optimal selection of MTDs for the fast uplink grant allocation. When the number of active MTDs is larger than the number of fast uplink grants that can be allocated, an optimal allocation policy must be developed. Therefore, to realize fast uplink grant allocation, the BS must have a mechanism to predict the set of active MTDs at each time. This is a type of traffic modeling known as \emph{source traffic modeling} \cite{MTCsourceTraffic} which is inherently different from the aggregate traffic modeling \cite{m2mtrafficstudy} at the BS. Source traffic prediction can be categorized into two groups, \emph{periodic traffic} and \emph{event-driven traffic}. Clearly, periodic traffic prediction is easier compared to event-driven traffic prediction. For periodic traffic, one can adopt calendar based pattern mining techniques \cite{periodicpattern}. For event-driven traffic, event detection and source traffic prediction mechanisms must be developed. In \cite{DI_letter}, an MTD traffic prediction method based on the so-called directed information is presented for source traffic prediction. By using the method proposed in \cite{DI_letter} upon detection of an irregular transmission, a set of future active MTDs facing the same event can be detected. The authors in \cite{brownPredictive} propose a predictive resource allocation scheme for event-driven MTC in which MTDs are physically located across a line where their traffic pattern can be predicted.

The second step after predicting the source traffic is the optimal allocation of the fast uplink grants. If the BS has full knowledge of the QoS requirements of all the MTDs, this task is rather trivial. However, in practice, due to privacy and security issues, as well as financial benefits of data for the application, the MTDs might not reveal the nature of the application to the BS. Moreover, the QoS requirements of the MTDs might change at different times, due to changes in channel quality between the MTDs and the BS and the presence of various applications that must send data through a single MTD. Therefore, the BS must perform fast uplink grant allocation, with limited or no prior knowledge about the QoS requirements of the MTDs, and use the information revealed to the BS after the transmission for future fast uplink grant allocation purposes. One natural tool for such a task is multi-armed bandit (MAB) theory which are essentially a class of reinforcement learning problems \cite{sutton1998reinforcement}. MABs have been previously used in other wireless communications problems (e.g., see \cite{MABsinSmallCells} where a review of applications of MABs in small cells is provided.) The authors in \cite{MABd2d} use MABs for channel selection in device-to-device (D2D) communications and in \cite{MABdist}, MABs are used for distributed user association in energy harvesting small cell networks. MAB is also proposed for multi-user channel allocation for cognitive radio networks in \cite{MABmutliusers}. All of the aforementioned works focus on using MAB in problems where it is not possible to get full information on the state of the system and resource allocation with no prior knowledge is required. The optimal fast uplink grant allocation is a similar problem where an optimal resource allocation is needed with no prior knowledge on the QoS requirements of the MTDs in the system. Therefore, MAB theory is a natural tool for our problem.

The main contribution of this paper is to address the problem of optimal fast uplink grant allocation with no prior information about QoS requirement of the MTDs using the MAB theory. We  consider that the BS is not able to perfectly predict the set of the active MTDs and hence, a probability of activity is associated with each MTD at any given time. Therefore, the BS has probabilistic knowledge on the set of active MTDs and we propose a novel MAB algorithm for allocating the fast uplink grant under these conditions. The contributions of this paper can, therefore, be summarized as follows:

\begin{itemize}
\item In order to capture a diverse set of QoS metrics during scheduling, we introduce a compound QoS metric that is a combination of three MTD-specific metrics: a) the value of the data packets, b) maximum tolerable access delay, and c) the data rate. We concretely define this metric by proposing a novel method to model the access delay by mapping it to a value between zero and one using a sigmoid function known as Gompertz function. 
\item To find the optimal MTD that the BS must schedule at each time slot, a novel probabilistic sleeping MAB algorithm is proposed. Sleeping MABs are appropriate to address the problems where the set of of active MTDs change over time. Our probabilistic version of the algorithm takes the probabilistic knowledge on the activity of the each MTD into account. The proposed algorithm combines the probability of activation of each arm with the concept of upper confidence bound (UCB) in the context of sleeping MABs. 
\item We rigorously analyze the regret of the proposed MAB algorithm and decouple the effect of the MTC source traffic prediction errors and the learning process on the regret. We analytically derive the conditions under which the errors in MTC source traffic prediction lead to selecting an MTD with lower utility value thereby increasing the regret of the proposed MAB algorithm.
\item Simulation results show that for any source traffic prediction algorithm with good accuracy, the proposed algorithm is optimal since it achieves logarithmic regret. For example, the proposed framework achieves up to three-fold improvement in the access delay compared to a baseline random scheduling policy.
\item We extend the proposed probabilistic sleeping MAB from single MTD selection to several MTD selection by using the concept of best ordering of bandits and provide an algorithm for scenarios where multiple MTDs can be scheduled at any given time. In this method, MTDs with highest UCB value are selected for transmission, which achieves much better performance in terms of delay and throughput compared to the baseline random allocation policy. Here, our simulation results show two-fold performance improvement in terms of latency compared to a baseline random allocation policy.
\end{itemize}

The rest of the paper is organized as follows. Section \ref{sysmodel} presents the system model and problem formulation. In Section \ref{MABsection}, we introduce the proposed probabilistic sleeping MAB solution and its extension to multiple MTDs and provide the regret analysis and study the effect of the source traffic prediction accuracy on the performance of the MAB algorithm. Numerical results are presented in Section \ref{Simulationresults} and conclusions are drawn in Section \ref{conclutions}.
\section{System Model and Problem Formulation}\label{sysmodel}
Consider the uplink of a cellular system composed of one BS and a set $\mathcal{M}$ of $M$ MTDs that use a fast uplink grant. Scheduling is done at the BS and a fast uplink grant is sent to each scheduled MTD. We assume that the total available bandwidth is divided into resource blocks, each of which is of size $W$ and duration $\tau$. Without loss of generality, we consider the problem of selecting one MTD for the fast uplink grant at each time duration $\tau$. Hereinafter, we use $i$ for indexing MTDs and $t$ for time. Due to the heterogeneous nature of IoT applications, packets are assumed to have different QoS requirements. The system model is presented in Fig. \ref{systemmodelfigure}.

\begin{figure}[t]
	\centering
	\includegraphics[width=0.5\textwidth]{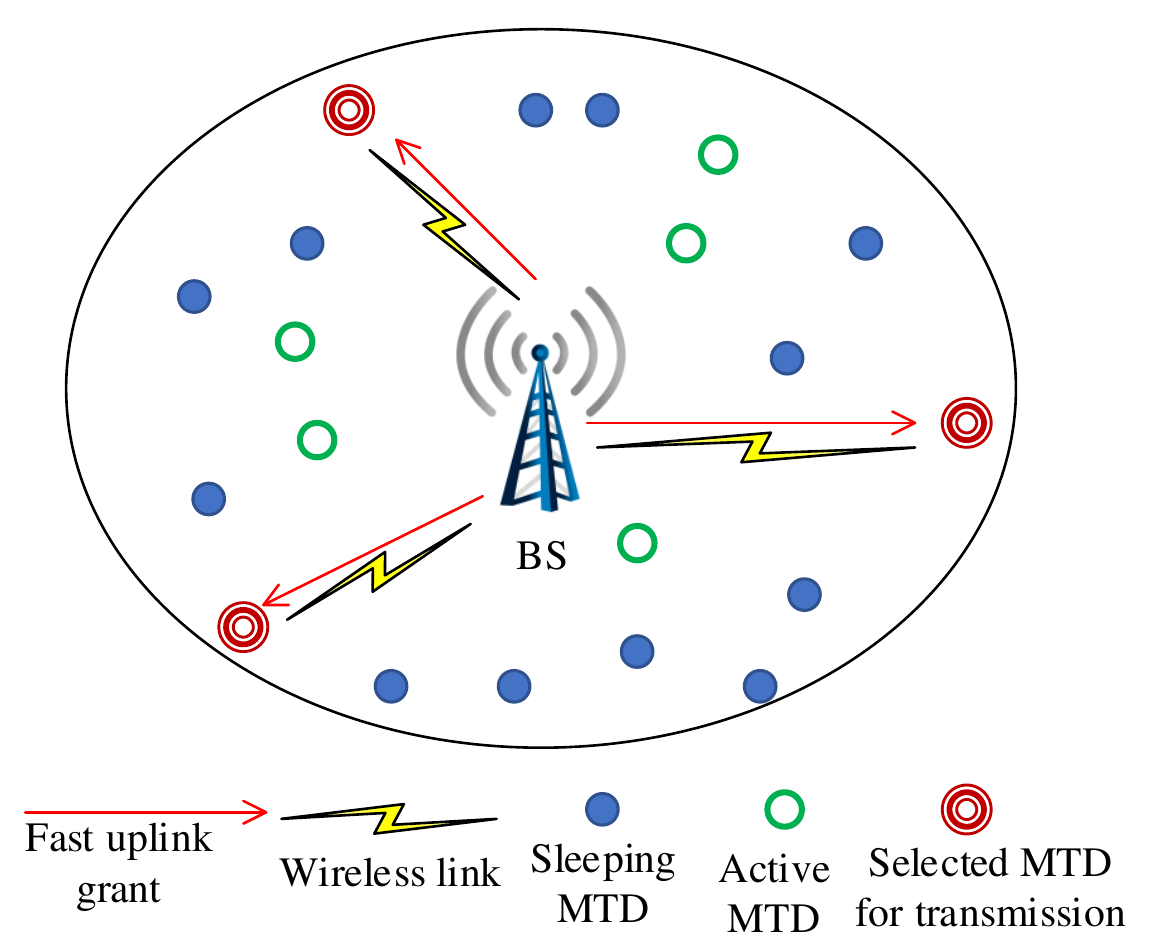}
	\caption{\small Illustration of system model. First, the set of active MTDs are predicted. Next, selected MTDs receive the fast uplink grants and transmit their data.}\label{systemmodelfigure}
\end{figure}

\subsection{Performance Metrics}
We now define three performance metrics that are combined to build a single metric that is used in the problem formulation.
\subsubsection{Value of information}
At time $t$, for each MTD $i$, we define the \emph{value of information} as the assessment of the utility of an information product in a specific usage context \cite{ValueofInformation}. Hence, each packet that arrives at the queue of an MTD $i$ will have an associated value $v_i(t)$. According to \cite{ValueofInformation}, this value can be determined by relative pairwise comparison of all IoT applications and the use of the so-called analytic hierarchy process (AHP) to calculate the importance weight for each packet. This normalized value is derived in the form of a percentage of importance, and hence we choose $v_i(t) \in [0,1]$.
\subsubsection{Maximum tolerable access delay}
Delay in a wireless communication network consists of different components: Processing delay $T_p$ which is a function of hardware and software used by the MTDs, queuing delay $T_q$, and transmission delay $T_t$ which pertains to the delay for the transmission of the data packets through the physical medium. Once the data is transmitted and received at the BS, the time needed for the packet to travel to the final destination through a network of wireless, wired, or fiber link is called routing delay $T_r$. Finally, the access delay $T_a$, which is the main focus of this work, is the time duration from the moment that the packet is ready for transmission, until the MTD receives the uplink resource blocks to transmit the packet. For each data packet of MTD $i$, we consider a \emph{maximum tolerable access delay} $d_i(t_s)$ defined as the total delay that can be tolerated from the time instance $t_s$ at which the data packet is ready to be transmitted at the MTD queue until it is scheduled to be sent. To calculate the total access delay that can be tolerated for each MTD, we first assume that all the other delay components are modeled and subtracted from the total tolerable delay of the packet. We assume $T_p$ and $T_t$ to be constant since the packets are small and always generated by the same devices, and the MTDs are either stationary or have low mobility. Since most of the MTDs have sparse packet transmissions and we can consider that the service time is considerably shorter than the packet inter-arrival times and, hence, the queuing time resulting from other packets in the system is considered to be negligible. Once all the delay components are modeled, we can calculate the maximum tolerable access delay as follows:
\begin{equation}
T_a = T_{\textrm{total}} - T_t - T_p.
\end{equation}
Due to the fact that the values of $T_{\textrm{total}}$, $T_r$, $T_t$, and $T_p$ are constant and that each application that is transmitting through the MTD might have different QoS requirements, the maximum tolerable access delay for each MTD will be different at any given time. Moreover, once a packet is in the MTD queue and waits for to access the channel, after each time step of waiting, its tolerable access delay will be shorter. Therefore, the packets of each MTD might have different tolerable access requirements at different times.

\subsubsection{Throughput}
Once each signal is received at the BS, the signal-to-noise ratio (SNR) is given by:
\begin{equation}
\gamma_i(t) =  \frac{q_i(t)|{h}_i(t)|^2}{WN_0},
\end{equation}
where $h_i(t)$ represent the channel between MTD node $i$ and the BS. $N_0$ is the power spectral density of the noise, $W$ is the bandwidth of the transmission channel, and $q_i(t)$ is the transmit power of MTD $i$. The channel is modeled as $h_i(t) = a_i(t).g_i(t)$ where $g_i(t)\sim\mathcal{CN}(0,1)$ represents the small-scale Rayleigh fading, assumed to be independent at different times \cite{mimoMatlab}. Large scale fading is included in $a_i(t) = 10^{\frac{a_{i,dB}(t)}{10}}$ where $a_{i,dB}(t) = PL_{dB} + X_\sigma$ with $PL_{dB}$ and $X_\sigma$ denoting the path loss and log-normal shadowing with variance $\sigma$. We use the 3GPP path loss model from the BS to MTDs \cite{3GPPPathLossModel} which is given by $PL_{\textmd{dB}} = 128.1 + 37.6\log(d)$. Subsequently, the rate is given by:
\begin{equation}
\begin{aligned}\label{eq:scSINR}
C_i(t) = W\log \bigg(1+ \frac{q_i(t)|{h}_i(t)|^2}{WN_0}\bigg).
\end{aligned}
\end{equation}
\subsection{Problem Formulation}
We first normalize $C_i(t)$ as well as the maximum tolerable access delay to a value within the range $[0,1]$. For the rate $C_i(t)$, we simply divide the achieved rate by the maximum rate $C_{\textrm{max}}$ that can be achieved by the node having the best channel to the BS. We fix $C_{\textrm{max}}$ for the entire period by using the knowledge of the set of all the MTDs that are registered in the network. Thus, we use a normalized rate $C_{i}^n(t) = C_i(t)/C_{\textrm{max}}$.

To normalize the maximum tolerable access delay, we use a mapping from maximum tolerable access delay to a number in $[0,1]$ using a function $f(d_i(t))$. To do this, we use Gompertz function \cite{Gompertz} with slight modifications, which is an asymmetric sigmoid function that is widely used in growth modeling. The rationale behind using this function is that it is possible to control the point at which the value of the function starts to decrease as well as the steepness of the curve. Gompertz function \cite{Gompertz} is given by $g(t) = ae^{-be^{-ct}},$ where parameter $a$ defines the asymptote of the function, $b$ sets the displacement along the time axis, and $c$ determines the growth rate or the steepness of the function. The Gompertz function is an increasing function in time. Moreover, since smaller values of the maximum tolerable access delay mean that the MTD has delay-sensitive data to transmit, and hence, it should have a higher value in the utility function, we modify the Gompertz function to create a new function that is decreasing with time, as follows:
\begin{equation}\label{modifiedgompetz}
f(d_i(t)) =  a - ae^{-be^{-cd_i(t)}}.
\end{equation}

Fig. \ref{gompetzfigure} shows the plot of the modified Gompertz function for some different values of the control parameters. Any scheduling algorithm performs better in terms of delay if it selects MTDs with smaller maximum tolerable access delay, which is the one that maximizes function $f(d_i(t))$.
\begin{figure}
	\centering
	\includegraphics[width=10cm]{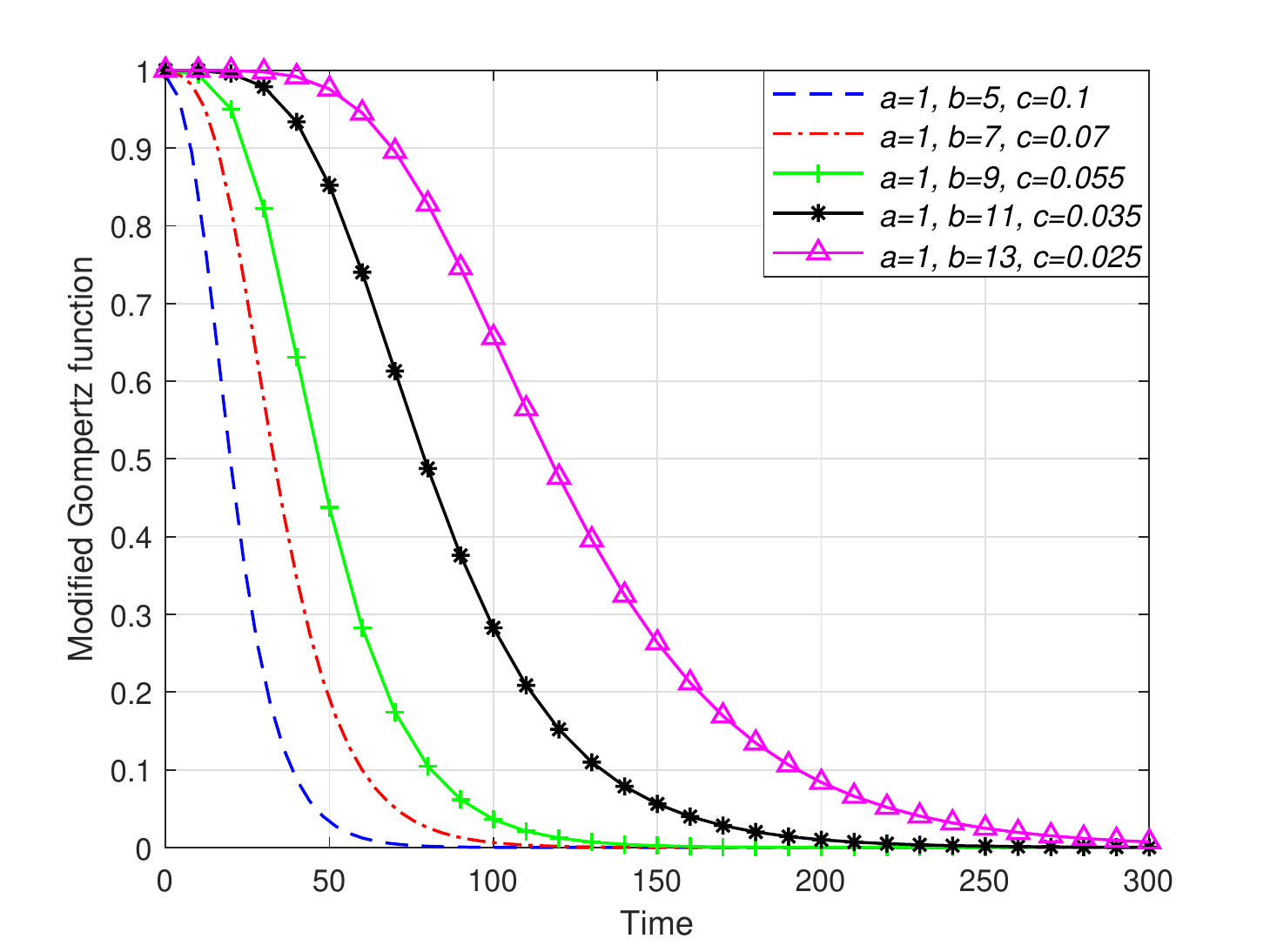}
	\caption{\small Modified Gompertz function for modeling latency for different values of the control parameters.}\label{gompetzfigure}
\end{figure}
For each MTD $i \in \mathcal{M}$, we can now define a utility function that combines all of the QoS metrics:
\begin{equation}\label{utilityfunction}
U_i(t) = \alpha v_i(t)  + \beta C_{i}^n(t)  + \gamma f(d_i(t)).
\end{equation}
In (\ref{utilityfunction}), $\alpha$, $\beta$, and $\gamma$ are weight parameters used to modify the importance of each metric with $\alpha+\beta+\gamma = 1$ . The best performance at time $t$ is achieved if an MTD $k \in\mathcal{K}\subseteq\mathcal{M}$ is selected such that:
\begin{equation}\label{optproblem}
\begin{aligned}
k=&\underset{i\in\mathcal{K}}{\text{argmax}}
& & U_i(t),\\
& \text{s.t.}
& & C_i(t) \geq \rho, \\
& & & d_i \geq t-t_s,
\end{aligned}
\end{equation}
where $\mathcal{K}$ is the set of active MTDs and $\rho$ is rate threshold required for data transmission. If $v_i(t)$, $h_i(t)$, $d_i(t)$, and the set of active MTDs are available to the BS, solving (\ref{optproblem}) is straightforward. However, in real-world networks, having such information at the BS is impractical due to the following reasons. First, MTDs should send a scheduling request to the BS using periodically available random access slots. Sending scheduling requests in MTC is not optimal since it: a) will most likely fail in massive access scenario, b) requires large signaling overhead compared to the small data packet size, and c) increases the latency. This motivates the development of a predictive resource allocation scheme, where the set of active MTDs is predicted at the BS. Second, for optimal performance in  the system, the BS must know the channel state information (CSI) of the MTDs, their data values, and their exact latency requirements. Clearly, in practical MTD networks, the BS does not have full knowledge on the parameters of the metric defined in (\ref{utilityfunction}). For example, since the data packets are small, having instantaneous CSI at the BS requires signaling overhead that is almost equal to the data size, which is naturally inefficient. Moreover, as discussed earlier, the tolerable access delay and value of the data packets can be different each time. Therefore, it is appropriate to solve problem (\ref{optproblem}) using online learning methods with limited or no information \cite{sutton1998reinforcement} at the BS. In this case the learning algorithm can learn the statistical properties of the CSI, the tolerable access delay, and the value of the data packets over time. Next, we propose a novel online algorithm based on MAB theory \cite{sutton1998reinforcement} to solve (\ref{optproblem}).

\section{Proposed Multi-Armed Bandit Framework and Algorithm}\label{MABsection}
\subsection{MAB theory and MAB problem formulation}
In a multi-armed bandit problem, a player (decision maker), pulls an arm from a set of available arms (selects an action from a set of available actions). Each arm generates a reward after being played, based on a distribution that is not known to the decision maker -- the decision maker only observes the reward of the selected arm. The aim of the player is to maximize a cumulative reward or minimize a cumulative regret. Regret is defined as the difference between the reward of the best possible arm at each game instant, and the generated reward of the arm that is played.

Let $\theta_{k}(t)$ be the reward of playing arm $k$ from the set of arms $\mathcal{K}$ at time $t$, and $\theta^*(t)= \underset{i\in \mathcal{K}}{\text{max}} \theta_{i}(t)$ to be the highest possible reward that could be achieved at time $t$ from the set of all arms $i\in \mathcal{K}$. The regret up to time $T$ is defined as \cite{sutton1998reinforcement}:
\begin{equation}\label{regret01}
R(T) = \mathbb{E}\bigg [ \sum_{t=1}^{T}\theta^*(t) - \sum_{t=1}^{T}\theta_{k}(t) \bigg ],
\end{equation}
where the expectation is taken over the random choices of the algorithm as well as the randomness in reward allocation. In our problem, each MTD is seen as an arm in the MAB settings and the BS is the player that selects the best arm at each time and after playing that arm, receives a reward that is generated by the metric defined in (\ref{utilityfunction}). Hence, the reward that is generated by each MTD $i \in \mathcal{M}$ is:
\begin{equation}\label{reward01}
\theta_i(t) = \mathds{1}\big[d_i > t-t_s\big] \mathds{1}\big[C_i(t)> \rho \big] U_i(t),
\end{equation}
where $\mathds{1}(.)$ is an indicator function that is equal to $1$ when the argument of the function holds and $0$ otherwise. Indicator functions are used to show that the reward of the algorithm at time step $t$ for selecting MTD $i$ is $0$ under the following conditions:
\begin{itemize}
	\item $C_i(t) < \rho $, i.e, the achieved throughput falls below the defined threshold and the packet cannot be transmitted successfully. This often happens when the channel quality between MTD $i$ and the BS is below a certain level.
	\item $d_i(t) < t_i - t_s$. Here, $t_i$ is the time that MTD $i$ is selected for transmission and $t_s$ is the time when MTD $i$ had a packet ready for transmission. Hence, $t-t_s$ is the number of time steps that MTD $i$ has waited to receive the fast uplink grant. Naturally, if $d_i(t) < t_i - t_s$, then the MTD packets will be dropped and the reward at the BS for selecting MTD $i$ will be $0$.
\end{itemize}

The goal of the BS is to maximize its cumulative reward over time. To solve such a problem, the natural solution is to find the best possible arm and play it all the time. This requires playing all of the available arms for many times to find their expected value. However, randomly selecting arms in the process of learning is highly suboptimal. Hence, an MAB algorithm finds the arms with higher rewards and chooses them more often, which is known as \emph{exploitation} of those arms. At the same time, an MAB algorithm should \emph{explore} all the other arms enough times to find their expected value more precisely. This is known as the exploration versus exploitation tradeoff. Several methods exist to solve the problem of exploration/exploitation. One of the most popular solution approaches for the MAB problem is based on the concept of upper-confidence bound (UCB). In this method, the MAB algorithm at each time $t$ plays an arm $x(t)$ such that:
\begin{equation}\label{ucb}
x(t) = \arg\max_{i \in \mathcal{K}} \frac{z_i(t)}{n_i(t)} + \sqrt{\frac{\psi \ln t}{n_i(t)}}
\end{equation}
where $t$ is the time step, $n_i(t)$ is the number of the times that arm $i$ was played in the previous time steps up to $t-1$, $z_i(t)$ is the sum of the rewards of playing arm $i$ up to time $t$, and $\psi$ is a parameter that provides a tradeoff between exploration and exploitation. Larger values of $\psi$ lead to a higher amount of exploration. We will next use the UCB concept in our proposed probabilistic sleeping MAB algorithm to provide a tradeoff between exploration and exploitation. In the UCB method, an interval is defined around the average of the received rewards from each arm. This confidence interval $\sqrt{\frac{\psi \ln t}{n_i(t)}}$ depends on the number of the times that an arm was played and the total number of the times that algorithm is running. The more one arm is played, the UCB value becomes smaller. This means that the empirical mean is closer to the real expected value of the arm. When using the concept of confidence intervals in MABs, the concept of \emph{optimism in the face of uncertainty} is used. This concept favors selecting an arm that has higher UCB value. 
\subsection{Sleeping Bandits and Proposed Algorithm}
In classical MAB problems, it is assumed that all of the arms are available to be played at all time instants. However, for the MTC fast uplink grant scheduling problem, this assumption is not valid since MTDs will have a small number of packets and usually, after each transmission, they become idle for some time. Hence, we consider a scenario in which, the set of available arms varies over time. This type of problems are called \emph{sleeping MAB} problems \cite{kleinberg2010regret}. In our problem, since the availability of the MTDs follows the distribution of their traffic, and the reward can be described by (\ref{reward01}), we have \emph{sleeping bandits with stochastic action availability and stochastic rewards}. The authors in \cite{kleinberg2010regret} provide an algorithm named AUER that addresses such problems and achieves optimal regret. However, AUER is only applicable to sleeping MAB problems in which the set of available arms is perfectly known to the decision maker in advance. In our problem formulation, such an assumption will not hold. Therefore, we propose a novel solution, summarized in Algorithm \ref{alg:probabilisticAlgoirthm}. Here, we consider that the BS has a prediction algorithm (e.g., such as those proposed in \cite{MTCsourceTraffic}, \cite{DI_letter}, and \cite{MingzheSurvey}) to determine the set of active MTDs at each given time. This algorithm provides the set of active MTDs with a certain probability. That is, each MTD $i$ has a probability $P_i(t)$ of being active at time $t$. In this problem, since the availability of the MTDs is probabilistic, the selected MTD might not be active, which will lead to $0$ reward and a waste of resources. Therefore, to solve the optimization problem in (\ref{optproblem}) we propose an MAB algorithm that takes such a probability of being active into account. In this algorithm, the BS at each time selects an MTD $x(t)$ such that:
\begin{equation}\label{maximization}
x(t)  = \arg \max_{i\in \mathcal{K}_{t}}  P_i(t) \left( \frac{z_i(t)}{n^'_i(t)} + \sqrt{\frac{\psi\log t^'}{n^'_i(t)}}\right)
\end{equation}
where $z_i(t)$ is the sum of rewards of MTD $i$, $n^'_i(t)$ is the number of the times that MTD $i$ was selected and was active, and $t^'$ is the total number of the times that the selected MTD was active. $\mathcal{K}_{t}$ is defined as the set of active MTDs at time $t$. In contrast to the original UCB method, we only count the number of times that the selected MTD was active. This ensures that the statistical average and the UCB values are calculated correctly. Since the availability of the MTDs in set $\mathcal{K}_t$ have associated probabilities, the error of the prediction at the BS will propagate to the MAB. This means that the performance of the sleeping MAB will suffer since some selected MTDs for the fast uplink grant might not be active. Less error in the prediction algorithm will lead to a better performance of the probabilistic sleeping MAB. This algorithm will select MTDs with higher values of the utility function and higher probability of being active while balancing the tradeoff between exploration and exploitation.
\begin{figure}[t]
		\begin{algorithm}[H]\small
			\caption{The Probabilistic Sleeping MAB Algorithm.}
			\begin{algorithmic}
				\State Initialize $z_i$, $n_i$ for all $i\in [n]$, initialize $t^'$
				\State \textbf{for} $t=1$ \emph{to} $T$ \textbf{do}
				\State \quad \textbf{if} $\exists j \in \mathcal{K}_t$ s.t. $n_j=0$ \textbf{then}
				\State \quad \quad Play arm $x(t)=j$
				\State \quad \textbf{else}
				\State \quad \quad Play arm $x(t)  = \arg \max_{i\in \mathcal{K}_{t}} P_i(t) \left( \frac{z_i(t)}{n^'_i(t)} + \sqrt{\frac{\psi\log t^'}{n^'_i(t)}}\right)$
				\State \quad \textbf{end}
				\State \quad \textbf{if} $x(t)$ is an available arm $(x(t) \neq 0)$ \textbf{then}		
				\State \quad observe payoff $\theta_{x(t)}$
				\State \quad $z_{x(t)} \leftarrow z_{x(t)} + \theta_{x(t)}(t)$
				\State \quad $n^'_{x(t)} \leftarrow n^'_{x(t)} + 1$
				\State \quad $t^' \leftarrow t^' + 1$
				\State \quad \textbf{else}
				\State \quad $z_{x(t)} \leftarrow z_{x(t)}$
				\State \quad $n^'_{x(t)} \leftarrow n^'_{x(t)}$
				\State \quad $t^' \leftarrow t^'$
				\State \textbf{end}
			\end{algorithmic}
			\label{alg:probabilisticAlgoirthm}
		\end{algorithm}
	\par
\end{figure}
\subsection{Prediction error}
Any error in the source traffic prediction algorithm that provides the set of active MTDs will affect the performance of the proposed sleeping MAB algorithm. Here, we first define the prediction error which will later be used in our analysis in Section \ref{regretAnalysis}. Two prediction errors:
\begin{enumerate}
	\item For any MTD that is active at time $t$ for which the source traffic prediction algorithm assigns a probability of being available $P_i(t)$, the prediction error will be $1-P_i(t)$. If an optimal MTD is active and has high prediction error, that MTD might not be scheduled and some sub-optimal MTD $j$ will be scheduled instead, which will lead to regret $\mu_{i}(t) - \mu_{j}(t)$. We use $e_1$ to capture this event.
	\item For any non-active MTD $j$ that is in the set $\mathcal{K}_t$, the prediction error is $P_j(t)$. If any non-active MTD $j$ is improperly selected due to high prediction error instead of an optimal MTD $i$, then the returned reward is zero, and, hence, regret is $\mu_{i}(t)$. This is the highest amount of regret that can happen at any given time. We denote event this by $e_2$.
\end{enumerate}
\subsection{Regret Analysis of the Proposed Algorithm}\label{regretAnalysis}
Next, we provide the analytical regret analysis of the proposed probabilistic sleeping MAB. We derive the upper bound of the regret and derive the relation between the accuracy of the source traffic prediction method and the regret of our proposed algorithm. Throughout this section, we use the following setup. Consider a MAB scenario with $n$ arms, where $ \mu_1>\mu_2> ... > \mu_n $, with $\mu_i$ being the expected value of the rewards of arm $i$. We define $ N_{i,j} $ as the number of times arm $ j $ was played while some arm in set $\mathcal{I}=\{1,\dots,i\} , (i<j) $ could have been played. We define $ \Delta_{i,j}=\mu_i-\mu_j $, which is always positive. The expected value of the regret can be expressed as:
\begin{equation}\label{eq:TheoremRegret}
\begin{aligned}
R(T)=&\mathbb{E}\left[\sum_{j=2}^{n}\sum_{i=1}^{j-1}\left(N_{i,j}-N_{i-1,j}\right)\Delta_{i,j}\right] + \mathbb{E}\left[\mu_{1}(t)\right]f(e_2)T\\
=&\sum_{j=2}^{n}\sum_{i=1}^{j-1}N_{i,j}(\Delta_{i,j} - \Delta_{i+1,j}) + \mathbb{E}\left[\mu_{1}(t)\right]f(e_2)T.
\end{aligned}
\end{equation}
$N_{0,j} = 0$ and $\Delta_{j,j} = 0$ \cite{kleinberg2010regret}. In the following, $n_{i}(t)$ is the number of times that arm $i$ is played until time $t$, $n^'_{i}(t)$ is the number of the times that arm $j$ was played and it was available, and $t^'$ is the number of the times that the played arm was available, $t$ is the time step in the algorithm and $T$ is the total time that the algorithm has been running. Moreover, in the following, $\hat{\mu}_{k}(t)$ shows the average received reward of arm $k$ up to time $t$. Next, we derive the number of times that prediction error event $e_2$ happens with function $f(e_2)$.
\begin{lemma}\label{Lemma1}
Given the definitions of $\mu_{k}$, $\hat{\mu}_{k}(t)$, and $n_{k}(t)$, the following holds:
\begin{equation}\label{eq:Lemma1}
\begin{aligned}
\mathbb{P} \left[\hat{\mu}_{k}(t) - \sqrt{\frac{\psi \ln t^'}{n^'_{k}(t)}}\leq \mu_{k} \leq \hat{\mu}_{k}(t) + \sqrt{\frac{\psi \ln t^'}{n^'_{k}(t)}}\right] = &\\ 
\mathbb{P} \left[\mu_{k} - \sqrt{\frac{\psi \ln t^'}{n^'_{k}(t)}} \leq \hat{\mu}_{k}(t) \leq\mu_{k} + \sqrt{\frac{\psi \ln t^'}{n^'_{k}(t)}}\right] &\geq 1 - \frac{2}{t^{2\psi}}
\end{aligned}
\end{equation}
\end{lemma}
\begin{proof}
We start from Chernoff-Hoeffding inequality where $\hat{\mu_{k}}(t)$ are strictly bounded by the intervals $[0, 1]$ and considering the confidence bound $\sqrt{\frac{\psi \ln t^'}{n^'_{k}(t)}}$, the inequality can be given by
\begin{equation}
\mathbb{P}\left[|\hat{\mu}_{k}(t) - \mu_{k}| \geq \sqrt{\frac{\psi \ln t^'}{n^'_{k}(t)}}\right] \leq 2 \exp\bigg(-n^'_{k}(t)\bigg(\sqrt{\frac{\psi \ln t^'}{n^'_{k}(t)}}\bigg)^2\bigg).
\end{equation}
After simplifications, we prove the lemma.
\end{proof}

This lemma is used in Theorem \ref{theorem1} where we analyze the regret bounds of the proposed probabilistic sleeping MAB solution presented in Algorithm \ref{alg:probabilisticAlgoirthm}. In our proposed MAB algorithm, a suboptimal arm is selected instead of the optimal arm in the following cases: a) The MAB algorithm does not have an accurate estimate of the rewards of each arm. This mostly happens during the initial learning phase, b) A suboptimal arm is selected because of prediction error $e_1$, or c) Zero reward is returned due to prediction error $e_2$. Clearly, cases a) and b) for the regret are a function of the accuracy of the prediction algorithm. We decouple the effect of the prediction errors of the source prediction algorithm from the uncertainty of the MAB algorithm about the expected values of the rewards of each MTD. We show that prediction errors can lead to linear regret with respect to the total running time of the algorithm with a coefficient that is a function of the prediction error. However, such a coefficient becomes very small for a source traffic prediction algorithm with high accuracy, and therefore, make the linear term very small.
\begin{theorem}\label{theorem1}
The regret of the probabilistic sleeping MAB algorithm is at most:
\begin{equation}\label{eq:The01}
\begin{aligned}
R(T)=&\bigg(4\psi \ln TP_{\textrm{av}}\!+\!\mathcal{O}(1)\!+\!f(e_1)T\bigg) \sum_{j=2}^{n}\sum_{i=1}^{j-1}\bigg(\frac{1}{\big(P_i\mu_{i}\!-\!P_j\mu_{j}\big)^2}\bigg)\bigg(P_i\mu_{i}\!-\!P_{i+1}\mu_{i+1}\bigg)^2\!+\!\mathbb{E}\left[\mu_{1}(t)\right]f(e_2)T\\
=&\bigg(8\psi \ln TP_{\textrm{av}} + \mathcal{O}(1) + f(e_1)T\bigg) \sum_{j=1}^{n-1}\bigg(\frac{1}{\big(P_{j+1}\mu_{j+1}-P_j\mu_{j}\big)^2}\bigg) + \mathbb{E}\left[\mu_{1}(t)\right]f(e_2)T.
\end{aligned}
\end{equation}
where $P_{\textrm{av}}$ is the average activity probability of the source traffic prediction.
\end{theorem}
\begin{proof}
To derive the regret bounds for our algorithm, we need to bound the regret arm by arm. We need to find the expected value of the number of the times that each arm was played, when that arm was suboptimal. That is, what is the expected number of times $N_{i,j}$ that arm $j$ was played while some other arm $i \in \mathcal{I},\ i \neq j$ could have been played. Assume that arm $j$ was already played $Q_{i,j}$ times while some other arm $i \in \mathcal{I}$ was available. The expected number of times that arm $j$ was played can be written as:
\begin{equation}\label{eq:proofstep1}
\begin{aligned}
N_{i,j} = &\sum_{t=Q_{i,j}+1}^{T} \sum_{s=Q_{i,j}+1}^{t} \textbf{}{\mathbb{P}}[(x_t=j) \wedge (j \text{ is played } s \text{ times)} \wedge (\mathcal{I} \neq \emptyset)]\\
&\leq \sum_{t=Q_{i,j}+1}^{T} \sum_{s=Q_{i,j}+1}^{t} \mathbb{P}\bigg[(x_t=j) \wedge (n^'_{j}(t) = s) \\& \wedge \bigg( \forall
_{k=1}^i\bigg(P_k(t)\hat{\mu}_{k}(t) + \sqrt{\frac{\psi \ln t^'}{n^'_{k}(t)}}\bigg) \leq \bigg(P_j(t)\hat{\mu}_{j}(t) + \sqrt{\frac{\psi \ln t^'}{s}}\bigg)  \bigg)\bigg] \\
&\leq \sum_{t=Q_{i,j}+1}^{T} \sum_{s=Q_{i,j}+1}^{t} \mathbb{P}\left[ \forall
_{k=1}^i\left(P_k(t)\left(\hat{\mu}_{k}(t) + \sqrt{\frac{\psi \ln t^'}{n^'_{k}(t)}}\right)\right) \leq P_j(t)\left(\hat{\mu}_{j}(t) + \sqrt{\frac{\psi \ln t^'}{s}}\right)\right]
\end{aligned}
\end{equation}
To analyze this, we define two events $E_1$ and $E_2$ as follows:
\begin{equation}
E_1 := \left[ \forall
_{k=1}^i\left(P_k(t)\left(\hat{\mu}_{k}(t) + \sqrt{\frac{\psi \ln t^'}{n^'_{k}(t)}}\right)\right) \leq P_j(t)\left(\hat{\mu}_{j}(t) + \sqrt{\frac{\psi \ln t^'}{n^'_{j}(t)}}\right)\right],
\end{equation}
and, for all $k\in\{j\}\cup \mathcal{I}$:
\begin{align}
E_2 := \hat{\mu}_{k}(t) \in \bigg[ \mu_{k} - \sqrt{\frac{\psi \ln t}{n_{k}(t)}},  \mu_{k} + \sqrt{\frac{\psi \ln t}{n_{k}(t)}} \bigg].
\end{align}
$E_2$ means that average received reward for each arm is not further away than the real value of expected value of the reward of each arm, within a margin of the UCB value. We have defined $E_2$ since it will help us in evaluating the accuracy of our estimation of the reward for each arm. After conditioning $E1$ on $E2$, we have:
\begin{align}
P[E_1] = &P[E_1|E_2]P[E_2] + P[E_1|E^c_2]P[E^c_2]\\
&\leq P[E_1|E_2] + P[E^c_2].
\end{align}
From Lemma \ref{Lemma1}, the probability of occurrence of $E_2$ for each arm is $1-1/t^{2\psi}$, and, thus, for all $k\in\{j\}\cup \mathcal{I}$ we have:
\begin{align}
P[E^c_2] = \frac{2(i+1)}{t^{2\psi}}.
\end{align}

No we can evaluate $E_1$ after conditioning on $E_2$. Event $E_1$ will happen if at least one of the following conditions hold \cite{auer2002finite}:

I) We are grossly overestimating the value of arm $j$:
\begin{align}
A_1 &:= P_j(t)\hat{\mu}_{j}(t) > \mu_j + \sqrt{\frac{\psi \ln t^'}{n^'_{j}(t)}}.
\end{align}
By carefully evaluating events $E_1$ and $E_2$, we can observe that $P[A_1|E_2] = 0$. Note that this overestimation is evaluated considering the worst case scenario with $P_j(t) = 1$.

II) We are grossly underestimating the values of all of the arms in $\mathcal{I}$, which can be captured by the following event:
\begin{align}
A_{2,1} := \forall
_{k=1}^i \bigg(P_k(t)\hat{\mu}_{k}(t) < \mu_{k} - \sqrt{\frac{\psi \ln t^'}{n^'_{k}(t)}} \bigg).
\end{align}
For $P_k(t) \simeq 1 $, this term never holds when conditioned on $E_2$, i.e, $P[A_{2,1}| E_2, P_k(t) \simeq 1] = 0$. However, for $P_k < 1 $, arm $k$ will be grossly underestimated under the following condition:
\begin{align}
A_{2,2} :=	\forall
_{k=1}^i \bigg (P_k(t) < \frac{\mu_{k} - \sqrt{\frac{\psi \ln t^'}{n^'_{k}(t)}}}{\hat{\mu}_{k}(t)} \bigg ).
\end{align}
This means that, for all arms in $\mathcal{I}$, the probability of being active (while the arm is actually active) so low that the probabilistic UCB value is lower than the real expected value of the arm. However, $A_{2,2}$ is not sufficient for incurring regret and another condition must hold for arm $j$: the probability of being active must be high enough such that its probabilistic UCB value is within the confidence interval around the real expected value, i.e., we must have:
\begin{align}
P_j(t)\hat{\mu}_{j}(t) > \mu_{j} - \sqrt{\frac{\psi \ln t^'}{n^'_{j}(t)}},
\end{align}
which leads to:
\begin{align}\label{ppppp_j1}
A_{2,3} :=	
P_j(t) > \frac{\mu_{j} - \sqrt{\frac{\psi \ln t^'}{n^'_{j}(t)}}}{\hat{\mu}_{j}(t)}.
\end{align}
Since for selecting the suboptimal arm $j$ both $A_{2,2}$ and $A_{2,3}$ must hold, we define the event:
\begin{align}
A_{2,4} = A_{2,2} \wedge A_{2,3},
\end{align}
and, thus, if $A_{2,4}$ occurs, a suboptimal arm $j$ might be played which lead to increase in the accumulated regret. We should state that $A_{2,4}$ is independent of $E_2$. 

III) The expected value of the arms $j$ and $k$ are nearly equal. When the expected values of two arms are close to each other, following two conditions will lead to choosing a suboptimal arm: a) whenever the confidence interval of the suboptimal arm is large and, hence, the suboptimal arm has higher UCB value compared to the optimal arm, or b) When the UCB value of the optimal arm is larger than the suboptimal, but the optimal arm has lower probability, and, therefore the suboptimal arm is selected. These two conditions can be expressed by:
\begin{equation}\label{thirdReason}
A_{3,1} := P_j\mu_{j} + 2 \sqrt{\frac{\psi \ln t^'}{n^'_{j}(t)}} > P_k\mu_{k}.
\end{equation}
After rearranging (\ref{thirdReason}), to choose the optimal arm, the following condition is needed for the confidence interval:
\begin{equation}\label{condtionOfThird}
\sqrt{\frac{\psi \ln t^'}{n^'_{j}(t)}} < \frac{P_k\mu_{k}-{P_j}\mu_{j}}{2}.
\end{equation}
Now, in order for the condition in (\ref{condtionOfThird}) to hold, we must play arm $j$ enough times to have an exact estimate of its value:
\begin{align}\label{q1}
Q^'_{i,j}> \bigg \lceil \frac{4\psi\ln T^'}{\big(P_k\mu_{k}-{P_j}\mu_{j}\big)^2} \bigg \rceil.
\end{align}
This means that, conditioned on $E_2$, after playing arm $j$ for $Q^'_{i,j}$ times, $A_{3,1}$ will never happen since the confidence intervals are small enough. Therefore we have $P[A_{3,1}|E_2] = 0$.

Now, we can write (\ref{eq:proofstep1}) as:
\begin{equation}
\begin{aligned}
N_{i,j} & \leq \sum_{t=Q_{i,j}+1}^{T} \sum_{s=Q_{i,j}+1}^{t}\bigg[P[A_1|E_2] +  P[A_{2,1}|E_2] + P[A_{2,4}|E_2] + P[A_{3,1}|E_2] + P[E^c_2]]\bigg].
\end{aligned}\label{eq:proofEvents}
\end{equation}

We have already seen that $P[A_1|E_2] = 0$, $P[A_{2,1}|E_2] = 0$, and $P[A_{3,1}|E_2] = 0$. Moreover, $P[A_{2,4}|E_2] = P[A_{2,4}]$ since the probability of an MTD being active is independent of the event $E_2$. As observed from Lemma \ref{Lemma1}, we have $P[E^c_2] = 2(i+1)/t^{2\psi}$, and since $E_3$ has occurred, (\ref{eq:proofEvents}) simplifies to:
\begin{align}
N_{i,j} &\leq Q^'_{i,j} + \sum_{t=Q_{i,j}+1}^{T} \sum_{s=Q_{i,j}+1}^{t}\bigg[ P[A_{2,4}]\bigg]  + \sum_{t=Q_{i,j}+1}^{T} \sum_{s=Q_{i,j}+1}^{t} \frac{2(i+1)}{t^{2\psi}}\\
&= \bigg \lceil \frac{4\psi\ln T^'}{\big(P_i\mu_{i}-{P_j}\mu_{j}\big)^2} \bigg \rceil +  \sum_{t=Q_{i,j}+1}^{T} \sum_{s=Q_{i,j}+1}^{t}\bigg[ P[A_{2,4}] \bigg] + \mathcal{O}(nT^{-\psi})\\
&= \bigg \lceil \frac{4\psi\ln T^'}{\big(P_i\mu_{i}-{P_j}\mu_{j}\big)^2} \bigg \rceil + \mathcal{O}(1) + \sum_{t=Q_{i,j}+1}^{T} \sum_{s=Q_{i,j}+1}^{t}\bigg[ P[A_{2,4}]\bigg] \bigg].
\end{align}
It is impossible to derive a closed-form expression for the number of times that event $A_{2,4}$ happens since the confidence interval and accuracy of the estimated average for each arm changes at each time. However, we can conclude that the number of times that $A_{2,4}$ happens is a linear function of time $T$ multiplied by a coefficient that is the function of the prediction error $f(e_1)$. Clearly, $f(e_1) \rightarrow 0$ as $e_1 \rightarrow 1$.
Therefore, we have:
\begin{equation}
N_{i,j}  = \bigg \lceil \frac{4\psi\ln T^'}{\big(P_i\mu_{i}-{P_j}\mu_{j}\big)^2} \bigg \rceil + \mathcal{O}(1) + f(e_1)T
\end{equation}

We can now exactly calculate $Q^'_{i,j}$. To this end, we need to count the total number of times that a given arm was selected that this arm and was available which is equal to the total number of times that we have selected an arm multiplied by probability of the availability of that arm, i.e., 
\begin{equation}
\begin{aligned}
\frac{n^'_j}{n_j}& =\mathbb{E}[P_j(t)] = P_j  \Rightarrow n^'_j = n_j P_j \quad \forall j \in \mathcal{A}, \\
\frac{t^'}{t} & =\mathbb{E}[P_j]  = P_{av} \Rightarrow t^' =  t P_{av}.
\end{aligned}
\end{equation}
Therefore, we can derive $N_{i,j}$ as:
\begin{align}
N_{i,j} & = \bigg \lceil \frac{4\psi\ln(T P_{av})}{ \big(P_i\mu_{i}-{P_j}\mu_{j}\big)^2} \bigg \rceil + \mathcal{O}(1) + f(e_1)T.
\end{align}
By plugging this in (\ref{eq:TheoremRegret}), we can conclude:
\begin{equation}\label{eq:TheFinal}
\begin{aligned}
R(T)=&\bigg(4\psi \ln TP_{\textrm{av}}\!+\!\mathcal{O}(1)\!+\!f(e_1)T\bigg) \sum_{j=2}^{n}\sum_{i=1}^{j-1}\bigg(\frac{1}{\big(P_i\mu_{i}\!-\!P_j\mu_{j}\big)^2}\bigg)\bigg(P_i\mu_{i}\!-\!P_{i+1}\mu_{i+1}\bigg)^2\!+\!\mathbb{E}\left[\mu_{1}(t)\right]f(e_2)T\\
=&\bigg(8\psi \ln TP_{av} + \mathcal{O}(1) + f(e_1)T\bigg) \sum_{j=1}^{n-1}\bigg(\frac{1}{\big(P_{j+1}\mu_{j+1}-P_j\mu_{j}\big)^2}\bigg) + \mathbb{E}\left[\mu_{1}(t)\right]f(e_2)T.
\end{aligned}
\end{equation}
This completes the proof.
\end{proof}
This theorem shows that the performance of the proposed sleeping MAB algorithm is a function of the accuracy of the predictions that are done in the previous step. This theorem shows that, for a source traffic prediction algorithm with good accuracy, after the learning period, the sleeping MAB will be able to select the most important MTD and it can achieve logarithmic regret. In MAB problems, logarithmic regret, as compared to linear regret shows that the algorithms has been able to learn the arms with higher reward and the gap between the selected arm and the best arm has become smaller \cite{sutton1998reinforcement}. In our MTC setting, this means that, our of the set of active MTDs, the one with best combination of latency requirements, wireless channel quality, and high value will be selected. Clearly, the we can change the coefficients of the reward that we have defined in (\ref{utilityfunction}) to give higher priority to the QoS of interest.
\subsection{Multiple MTD selection}
In the previous sections, we have studied the sleeping MAB algorithm for the fast uplink grant allocation problem. Most MAB algorithms are developed for selecting one arm at a time. However, in practical wireless systems, at any given time, there are multiple radio resources block that could be allocated to the MTDs, and hence, the network may need to select more than one MTD for resource allocation. Here, we extend the proposed sleeping MAB algorithm for multiple arms. We assume that there are $l$ radio resource blocks in the frequency domain that can be allocated for $l$ MTDs. In order to do this, since the criteria in selecting the best MTD in the probabilistic sleeping MAB algorithm was the MTD with highest UCB value, we extend our methods by selecting $l$ highest UCB values at each time step. This method follows the concept of best ordering of arms in MAB theory, in which, the arms are ordered based on their importance to be selected \cite{sutton1998reinforcement}. If we assume that the arms are selected one by one, after selecting the best MTD, for the next selection, we must choose the next arm with highest UCB value. Hence, the ordering of the UCB values and selecting the best $l$ MTDs is a very natural extension to the proposed probabilistic sleeping MAB. We should mention that at each time step, all of the MTDs that are active for the first time are selected first, and then other MTDs are sorted based on their UCB value. This proposed method of multiple MTD selection is summarized in Algorithm \ref{alg:mutipleResourceBlocks}. 
\begin{algorithm}[t]\small
	\caption{{Multiple MTD Selection}}
	\begin{algorithmic}
		\State Initialize $z_i$, $n_i$, $n^'_i$  for all $i\in [n]$, initialize $t^'$
		\State \textbf{for} $t=1$ to $T$ \textbf{do}
		\State \quad \textbf{if} $\exists i \in A_t$ \emph{s.t.} $n_i=0$ and $n^'_i=0$ \textbf{then}
		\State \quad \quad Play all arms with $x(t)=i$ or set $b$ = number of arms with $n_i=0$ and $n^'_i=0$
		\State \quad \textbf{else}
		\State \quad \quad Order the arms in descending oder by $P_i(t) \left(\frac{z_i}{n^'_i} + \sqrt{\frac{\psi\log t^'}{n^'_i}}\right)$ and select $(l - b)$ first arms
		\State \quad \textbf{end}
		\State \quad \textbf{if} $x(t)$ is an available arm $(x(t) \neq 0)$ \textbf{then}		
		\State \quad For all available arms, observe payoff $\theta_{x(t)}$
		\State \quad $z_{x(t)} \leftarrow z_{x(t)}$ + $\theta_{x(t)}$
		\State \quad $n^'_{x(t)} \leftarrow n^'_{x(t)} + 1$
		\State \quad $t^' \leftarrow t^' + 1$
		\State \quad \textbf{else}
		\State \quad for all non-available arms \textbf{do}
		\State \quad $z_{x(t)} \leftarrow z_{x(t)}$
		\State \quad $n^'_{x(t)} \leftarrow n^'_{x(t)}$
		\State \quad $t^' \leftarrow t^'$
		\State \textbf{end}
	\end{algorithmic}
	\label{alg:mutipleResourceBlocks}
\end{algorithm}
\section{Simulation Results}\label{Simulationresults}
\subsection{Single MTD selection}
We consider a single circular cell with radius $500$ meters consisting of $100$ MTDs with $10$ MTDs being active at each time. The noise power is considered to be $-174$ dBm/Hz and bandwidth is $360$ kHz and standard deviation for the log normal shadow fading is $10$ dB. All statistical results are averaged over a large number of independent runs. Each MTD has a reward distribution with expected value $U_i \in (0,1)$. The value of the reward function changes due to the following reasons. First, the achieved rate at each time changes due to changes in the channel quality. Second, the maximum tolerable access delay might change at different times since the packet in the MTD might face various delays. Moreover, each MTD can send packets from various applications with different data values. In the utility function, values $\alpha = 0.2$ $\beta = 0.3$, and $\gamma = 0.5$ are initially used. As needed, we change the parameters of the modified Gompertz function from Fig. \ref{gompetzfigure} based on the maximum access delay required in the system to have an accurate modeling of the latency.

In Fig. \ref{regretFigureResult}, we set $a = 1, b = 8$, and $c = 0.0.3$, and we show the regret resulting from the proposed sleeping MAB algorithm. Two intervals for the probabilities provided by a source traffic prediction are considered, one in $[0.8, 1]$ and another in $[0.9, 1]$. The result is compared to: a) A random scheduling policy, b) The case when the availability of the MTDs is not taken into account in the selection process of (\ref{maximization}) and only UCB values are used, and c) A scenario in which the prediction is error free. Fig. \ref{regretFigureResult} clearly shows that the random allocation of radio resources has linear regret which is much worse compared to the logarithmic regret achieved by the proposed solution. Fig. \ref{regretFigureResult} also shows that the proposed enhancement of our algorithm done by adding the probability in (\ref{maximization}) provides around $15\%$ and $50\%$ improvement on the performance compared to using the sleeping MAB without modification for two probability intervals. Moreover, Fig. \ref{regretFigureResult} shows that perfect prediction has the best performance. The baseline random policy performs very poorly in terms of regret as seen from Fig. \ref{regretFigureResult} since its regret increases linearly with time. 
\begin{figure}[t]
	\centering
	\includegraphics[width=10cm]{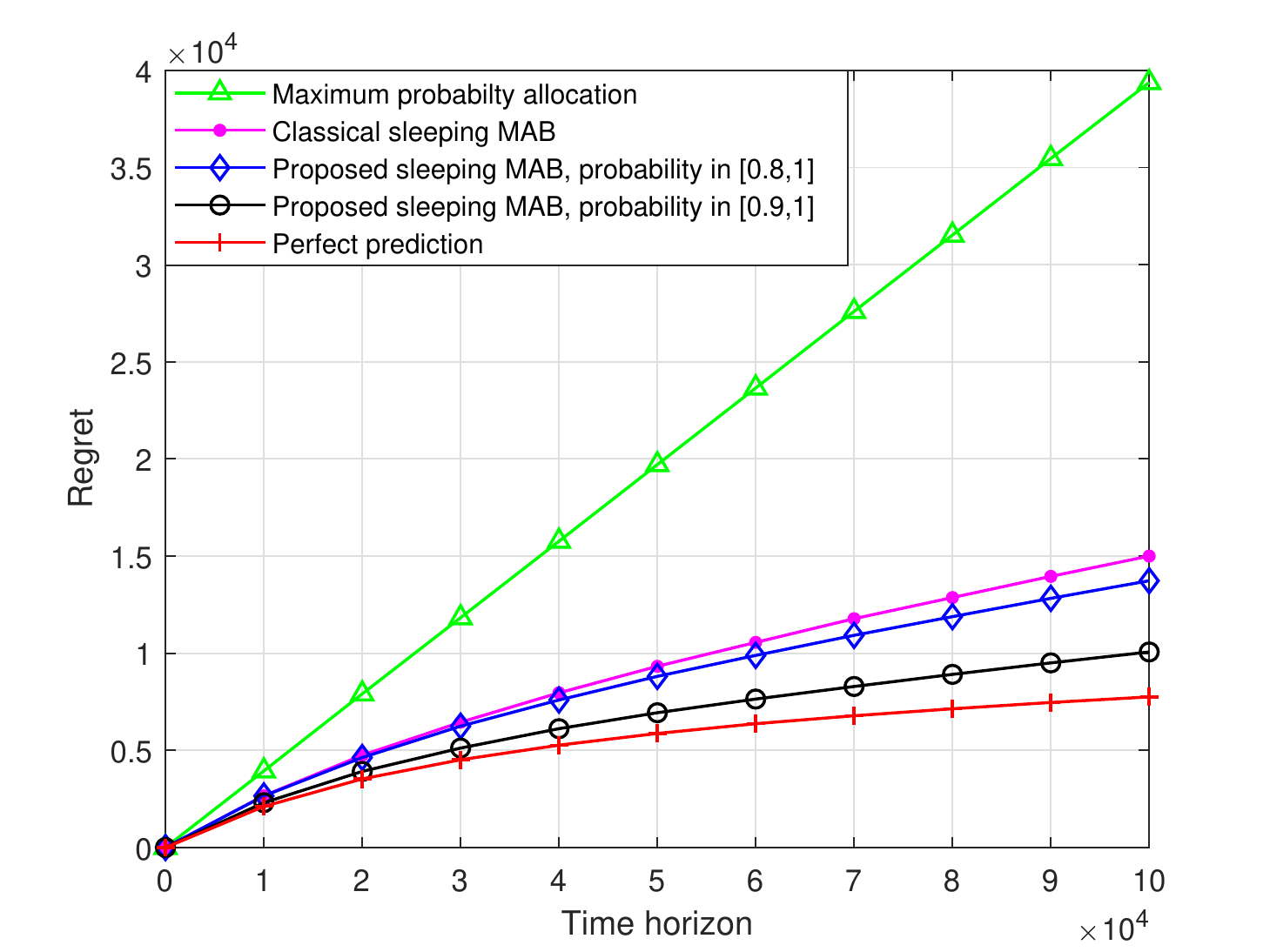}
	\caption{\small Regret resulting from the proposed probabilistic sleeping MAB compared to sleeping MAB with prediction, sleeping MAB with perfect prediction, and random allocation.}\label{regretFigureResult}
\end{figure}

In Fig. \ref{delayResult}, we consider $\alpha = \beta = 0$ and $\gamma =1$ to study the performance in terms of latency. The maximum tolerable access delay is considered to be a value in $[1, 300]$ ms and we set the parameters of the modified Gompertz function to $a = 1, b = 13$, and $c = 0.025$ with the time horizon $T = 10^6$. For every value of the maximum tolerable access delay in the system, the average maximum tolerable access delay of a random allocation policy is compared to the sleeping MAB algorithm. From Fig. \ref{delayResult}, we can see that the random allocation of the fast uplink grant achieves a delay that is equal to the average delay of the network. In contrast, the proposed algorithm is able to select MTDs with stricter latency requirements. The maximum tolerable access delay of the MTD selected by the proposed algorithm is almost three-fold less than that of the randomly selected MTD. Note that this scheduling policy not only decreases the average latency of the system but is also able to satisfy the individual latency requirements of each MTD by prioritizing the scheduling of MTDs with strict requirements.
\begin{figure}[t]
	\centering
	\includegraphics[width=10cm]{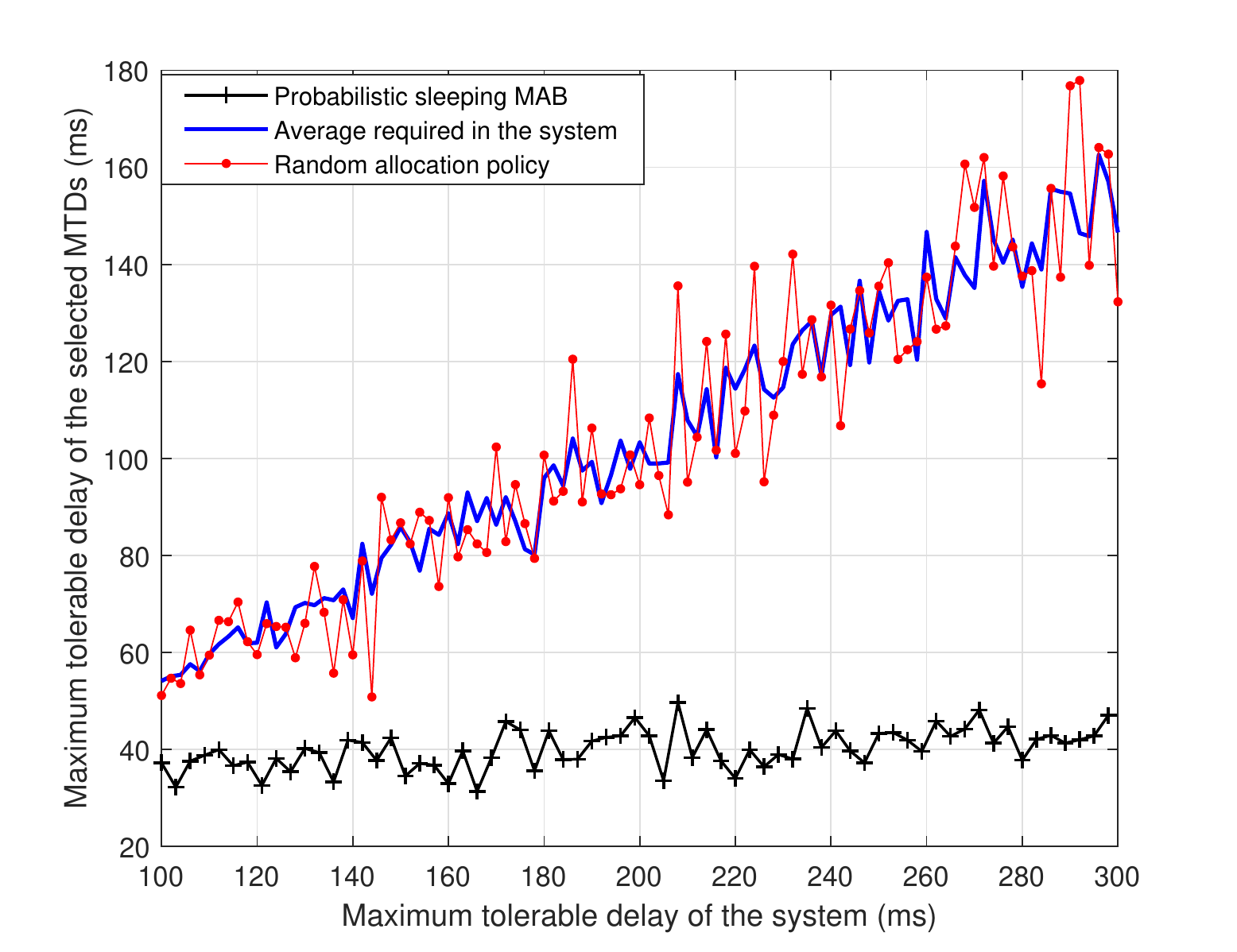}
	\caption{\small Average maximum tolerable access delay of the selected MTDs.}	\label{delayResult}
\end{figure}

The scatter plot of the latency of the selected MTD at each time is presented in Figs. \ref{scatter}(a), \ref{scatter}(b), and \ref{scatter}(c) for the proposed sleeping MAB for $\psi=1$, $\psi=6$, and $\psi=16$ respectively, and in Fig. \ref{scatter}(d) for the random allocation case. We set the maximum tolerable access delay to $100$ ms and the parameters of the modified Gompertz function to $a = 1, b = 7$, and $c = 0.07$. Each dot in these figures corresponds to the maximum tolerable access delay of the selected MTD. Fig. \ref{scatter}(a), \ref{scatter}(b), and \ref{scatter}(c) show the effectiveness of the sleeping MAB algorithm in optimizing the latency while providing fairness in the system. Those figures also show the effect of the explore/exploit control parameter $\psi$. From Figs. \ref{scatter}(a), \ref{scatter}(b), and \ref{scatter}(c), we can see that, initially, the dots are uniformly distributed which means that the MTDs are randomly selected. However, after learning, the intensity of the dots for MTDs with stricter latency requirements is much higher than that of the MTDs with larger delay requirement, which means that delay sensitive MTDs are scheduled more often. However, after after the learning period, the algorithm will keep scheduling MTDs with larger latency requirements. This increases the accuracy of the information at the BS about the latency requirements of all MTDs and also provides fairness. Moreover, if the latency requirements of an MTD has changed over time, the algorithm can discover that and start scheduling that MTD accordingly. Such a behavior shows how the proposed algorithm can balance between exploration and exploitation using parameter $\psi$. From Fig. \ref{scatter}(d), we can see that a random scheduling algorithm selects the latency completely randomly at all times and the performance of the system is much worse than the proposed sleeping MAB.

\begin{figure*}[t]
	\centering
	\subfigure[Probabilistic sleeping MAB with $\psi =1$]{\includegraphics[width=8cm]{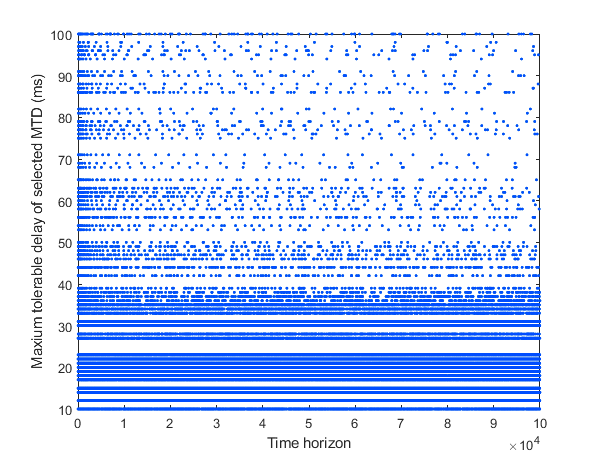}}
	\subfigure[Probabilistic sleeping MAB with $\psi =6$]{\includegraphics[width=8cm]{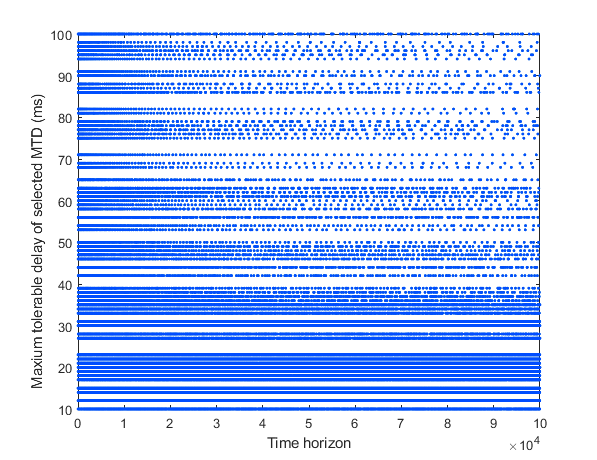}}
	\subfigure[Probabilistic sleeping MAB with $\psi =16$]{\includegraphics[width=8cm]{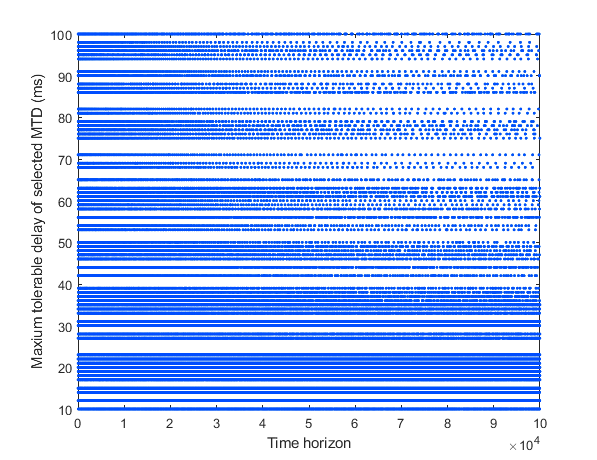}}
	\subfigure[Random allocation]{\includegraphics[width=8cm]{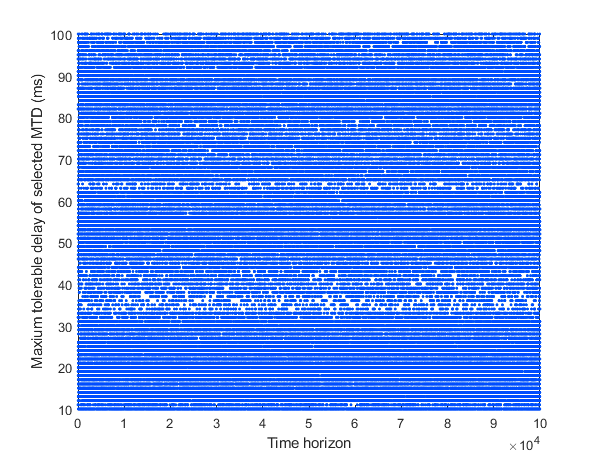}}
	\caption{\small Required access delay of the selected MTD at each time during the entire learning period. Effect of exploration-exploitation parameter $\psi$ is shown. This figure shows how our proposed method can optimize the system while providing fairness.}\label{scatter}
\end{figure*}
In Fig. \ref{th_scatter}, we present the scatter plot of the achieved throughput of the system at each time step for $\psi=0.5$, $\psi=2$, and $\psi=4$ and a random allocation policy. Here, we have set $\alpha = \gamma = 0$ and $\beta = 1$. The bandwidth is considered to be $360$ kHz and the transmit power of all the MTDs $10$ dBm. It is clear from Fig. \ref{th_scatter} that the random allocation policy, on average, achieves a lower rate. However, for $\psi=0.5$, the proposed method is showing much better average performance. 
\begin{figure*}[t]
	\centering
	\subfigure[Probabilistic sleeping MAB with $\psi =0.5$ ]{\includegraphics[width=8cm]{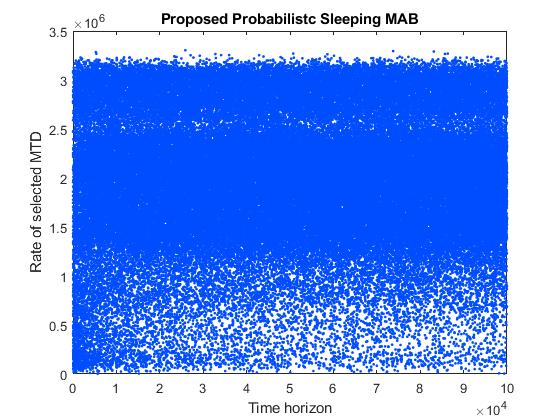}}
	\subfigure[Random allocation]{\includegraphics[width=8cm]{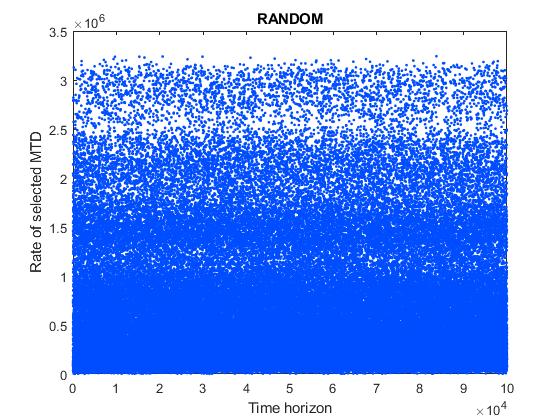}}\caption{\small Scatter plot of the achieved throughput of the system at each time step for three different values of the explore/exploit parameter compared to a random allocation policy.}\label{th_scatter}
\end{figure*}

In Fig. \ref{controller}, we present the average sum-rate of the system for the entire learning period for different values of the exploration-exploitation parameter $\psi$. Clearly, our proposed method outperforms the random allocation by up to $250 \%$. Fig. 7 shows that increasing $\psi$ decreases the capacity of the system. This means that scheduling the MTD that has higher average throughput can lead to to doubling the rate in the system. Increasing fairness is possible, but it will lead to a lower rate achieved in the system.
\begin{figure}[t]
	\centering
	\includegraphics[width=9cm]{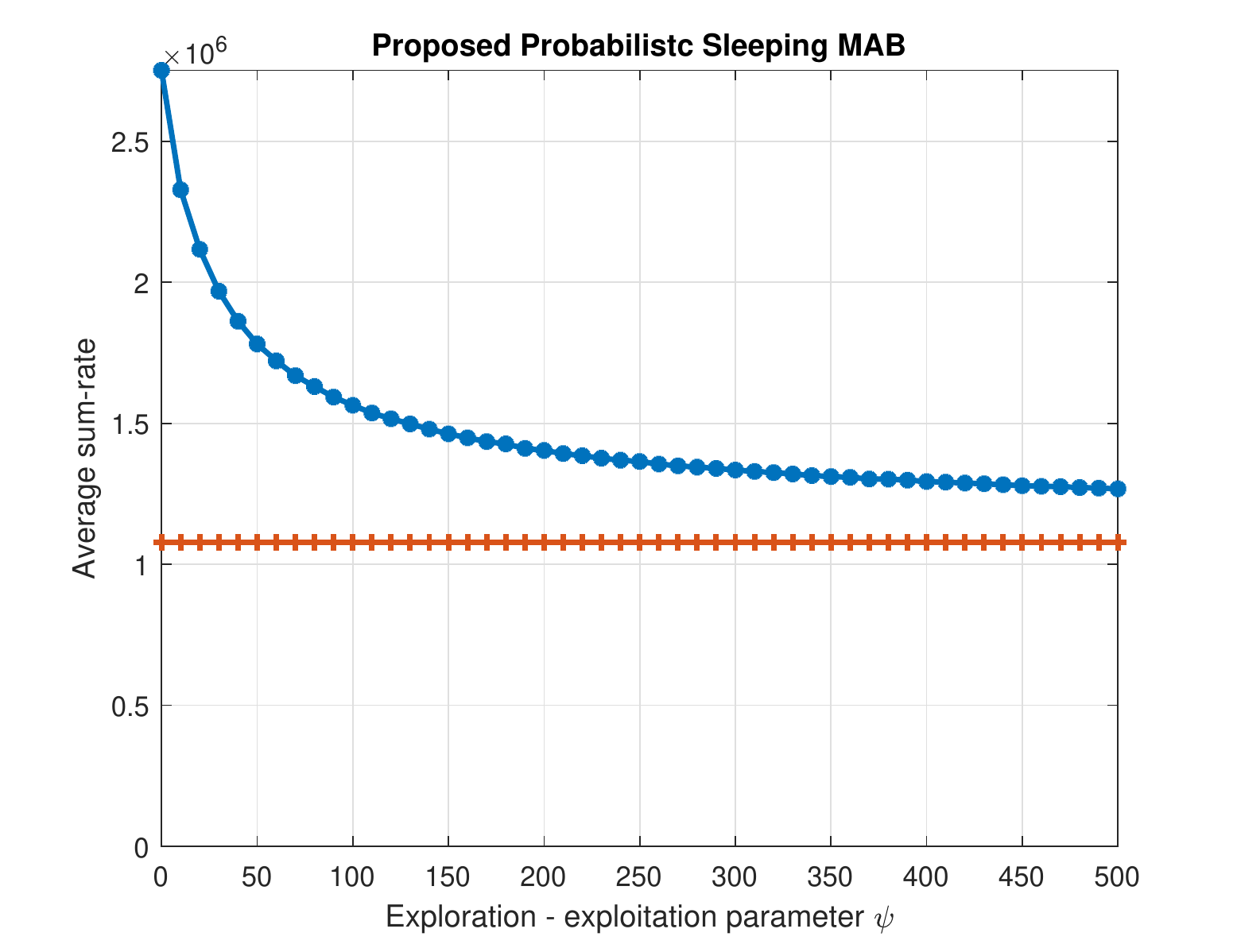}
	\caption{\small Average sum-rate of the system under different values of exploration-exploitation parameter $\psi$.}\label{controller}
\end{figure}
\subsection{Multiple Resource Blocks}
In this section, we provide the results for selecting multiple MTDs by using Algorithm \ref{alg:mutipleResourceBlocks}. Here, we consider that all the devices require the same amount of resources and one resource block is enough for transmitting the packet of each MTD. For each failed transmission, we consider the device to be available in the next time step. 

First, we provide the regret of the algorithm to study its performance. We set $a = 1, b = 8$, and $c = 0.0.3$, and the utility function values $\alpha = 0.2$ $\beta = 0.3$, and $\gamma = 0.5$.  We assume that there are $500$ MTDs in the system and, at each time, $50$ MTDs are active and $20$ MTDs can be scheduled at each time. We consider two intervals for the probability of being active that is provided by a source traffic prediction algorithm, one in $[0.8, 1]$ and another in $[0.9, 1]$. The regret of the proposed probabilistic sleeping MAB is compared to the random baseline scenario and perfect prediction. Fig. \ref{multipleRegret} shows that the proposed method achieves logarithmic regret. We observe that compared to the random allocation policy, the regret achieved by our proposed probabilistic sleeping MAB is nearly three and four times lower for the two different probability intervals that we considered. Fig. \ref{multipleRegret} naturally confirms that the perfect prediction scheme achieves the best performance.

\begin{figure}[t]
	\centering
	\includegraphics[width=9cm]{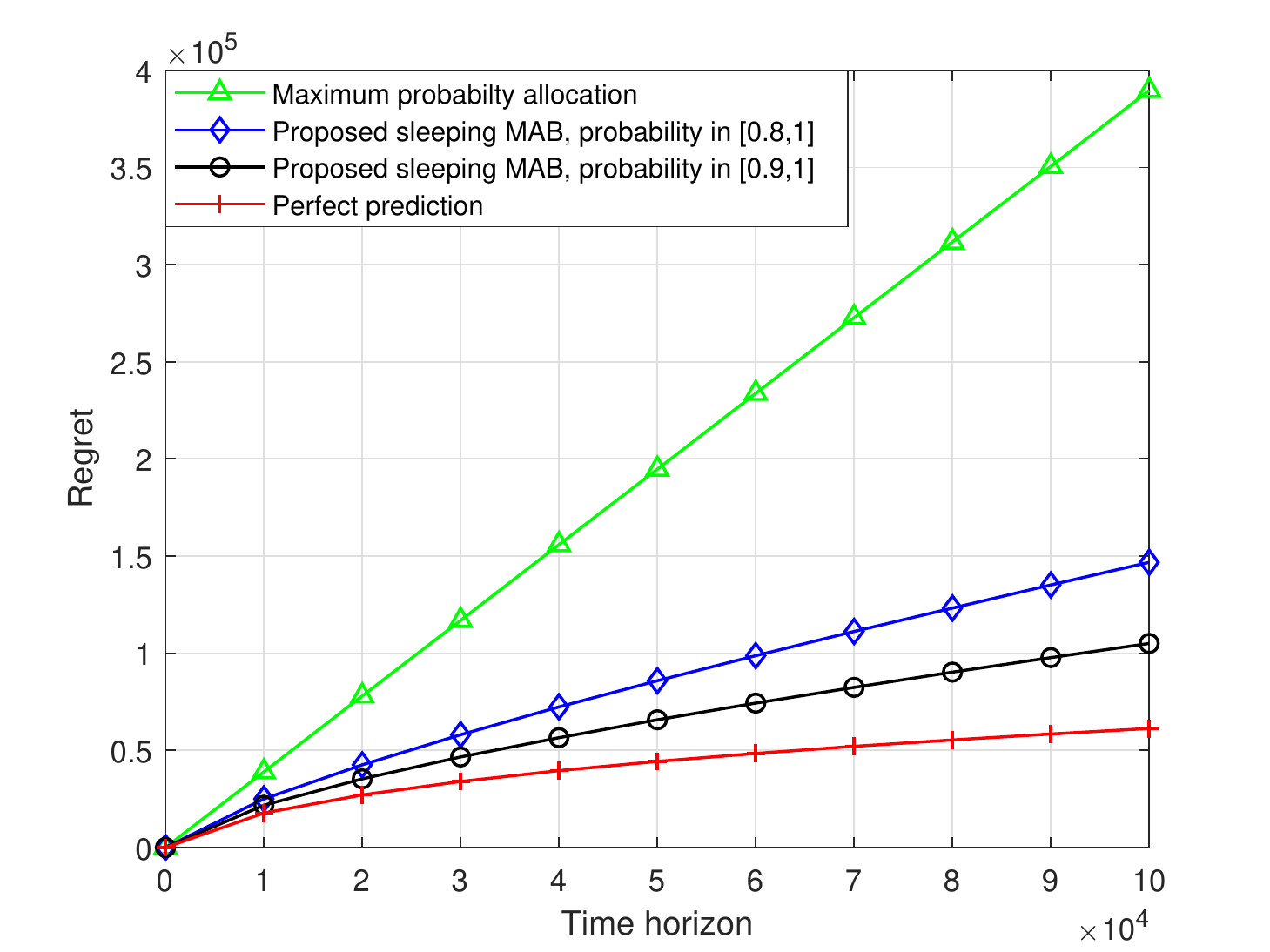}
	\caption{\small Regret resulting from the proposed probabilistic sleeping MAB compared to sleeping MAB for multiple resource blocks with prediction, sleeping MAB with perfect prediction, and random allocation.}\label{multipleRegret}
\end{figure}
In Fig. \ref{multiple_delay}, we present the average delay of the selected MTDs with $\alpha = \beta = 0$ and $\gamma =1$. The maximum tolerable access delay is considered to be a value in $[1, 300]$ ms and we set the parameters of the modified Gompertz function to $a = 1, b = 13$, and $c = 0.025$ with the time horizon $T = 10^6$. It is clear that the proposed probabilistic sleeping MAB algorithm is able to provide much better average achieved access delay in the system. One must note that the achieved average access delay is almost constant for any value of the maximum tolerable delay, since the select MTDs are averaged. This shows that, in real-time systems, by increasing the number of MTDs, our proposed solution achieves almost a three-fold improved performance compared to baseline methods. This is an interesting result since it shows that, for a massive access scenario, our proposed method is able to achieve very low access delay. In contrast, conventional random access based systems experience excessive delays due to collisions. 

\begin{figure}[t]
	\centering
	\includegraphics[width=10cm]{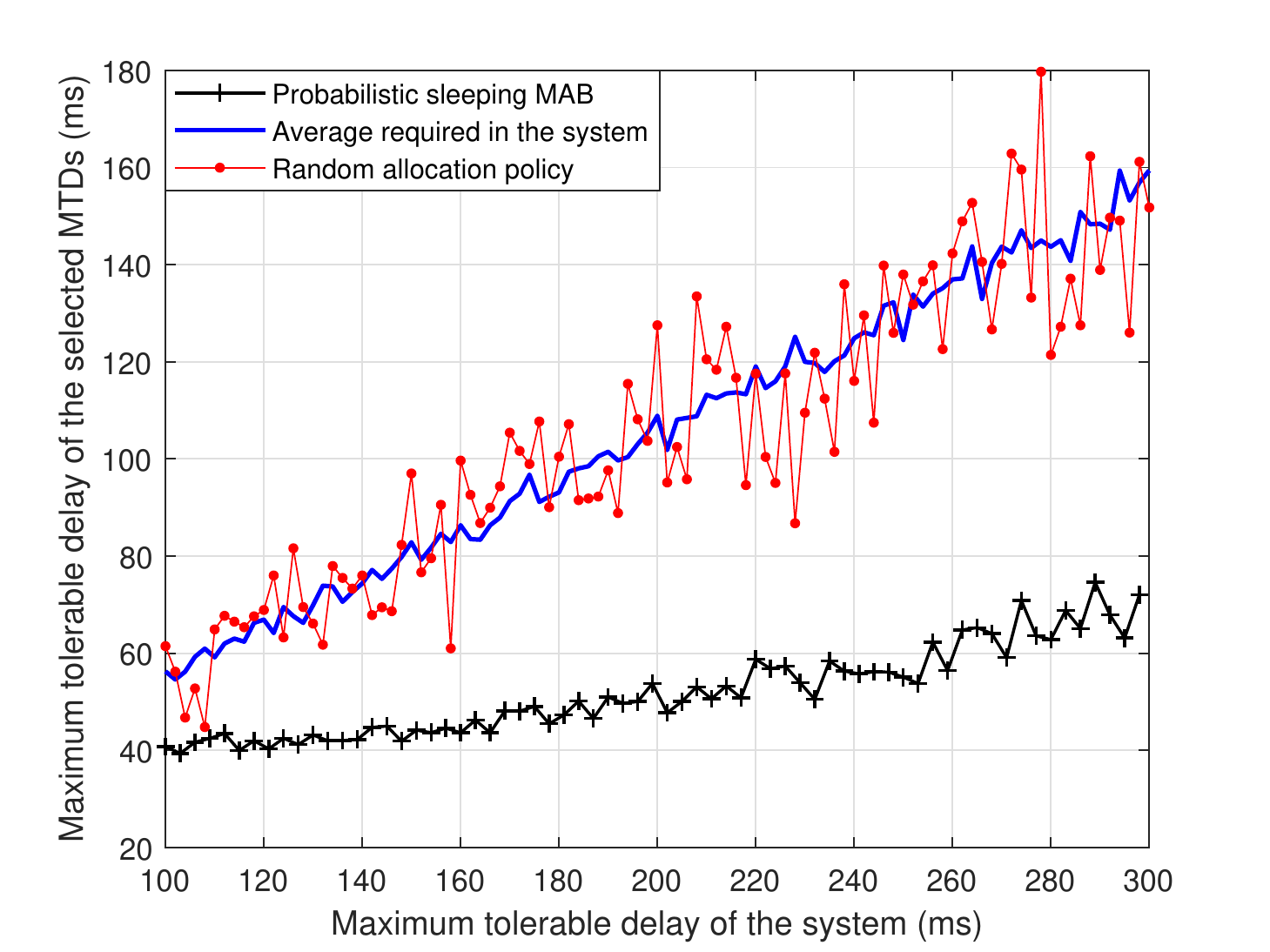}
	\caption{\small Average achieve access delay in the system for all the selected MTDs.}\label{multiple_delay}
\end{figure}

\begin{figure*}[t]
	\centering
	\subfigure[Probabilistic sleeping MAB]{\includegraphics[width=8cm]{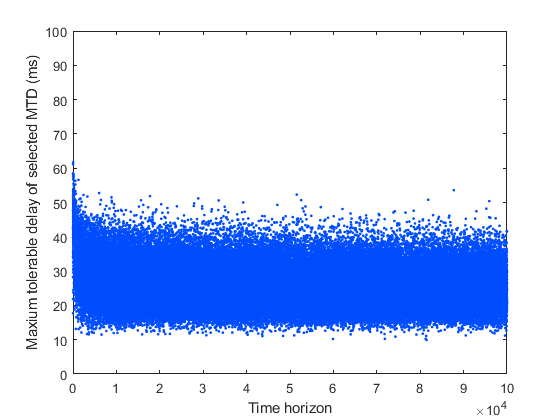}}
	\subfigure[Random allocation]{\includegraphics[width=8cm]{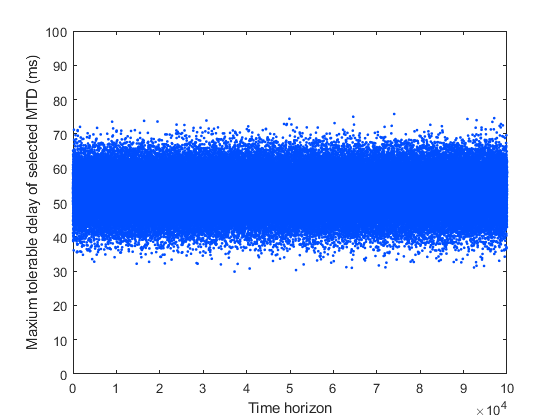}}\caption{\small Scatter plot of the average access delay of the system.}\label{multiple_delay_scatter}
\end{figure*}
In Fig. \ref{multiple_delay_scatter}, a scatter plot of the average access delay of the selected MTDs is presented for our proposed solution compared to a random allocation policy.  We set the maximum tolerable access delay to $100$ ms, and the parameters of the modified Gompertz function to $a = 1, b = 7$, and $c = 0.07$. Each dot in \ref{multiple_delay_scatter} captures the average of the maximum tolerable access delay of the selected MTDs. From \ref{multiple_delay_scatter}, we can clearly observe that the proposed sleeping MAB achieves a better performance and can improve the average latency in the system. Moreover, there is a balance between selecting the MTDs with the most strict access delay requirements and exploring other MTDs so as to provide fairness, which can be done by changing $\psi$.

\section{Conclusions}\label{conclutions}
In this paper, we have introduced a novel sleeping MAB framework for optimal scheduling of MTDs using the fast uplink grant. First, we have devised a mixed QoS metric based on a combination of the value of the data, rate of the link, and maximum tolerable access delay of each MTD. Second, we have used that metric as a reward function in a MAB framework whose goal is to find the best MTD at each time for scheduling. Moreover, we have considered an imperfect source traffic prediction where each MTD in the set of active MTDs has a probability of being active. Then, we have proposed a probabilistic sleeping MAB framework to solve the problem of fast uplink grant allocation. We have analytically studied the regret of the proposed probabilistic sleeping MAB and we have shown how the errors in the source traffic prediction algorithm impact the performance of the proposed sleeping MAB compared to the case of perfect source traffic prediction. Moreover, we have extended the sleeping MAB algorithm for selecting multiple arms at each time to use it in scenarios where more than one MTD is scheduled at each time. Simulation results have shown that the proposed algorithm performs much better than a random allocation policy, and can achieve almost three-fold performance gain in terms of latency and throughput. To the best of our knowledge, this is the first paper that addresses the optimal allocation of the fast uplink grant for MTC.
\bibliographystyle{IEEEtran}
\bibliography{bib_ref_mab_jrnl}

\begin{thebibliography}{10}
\providecommand{\url}[1]{#1}
\csname url@samestyle\endcsname
\providecommand{\newblock}{\relax}
\providecommand{\bibinfo}[2]{#2}
\providecommand{\BIBentrySTDinterwordspacing}{\spaceskip=0pt\relax}
\providecommand{\BIBentryALTinterwordstretchfactor}{4}
\providecommand{\BIBentryALTinterwordspacing}{\spaceskip=\fontdimen2\font plus
\BIBentryALTinterwordstretchfactor\fontdimen3\font minus
  \fontdimen4\font\relax}
\providecommand{\BIBforeignlanguage}[2]{{%
\expandafter\ifx\csname l@#1\endcsname\relax
\typeout{** WARNING: IEEEtran.bst: No hyphenation pattern has been}%
\typeout{** loaded for the language `#1'. Using the pattern for}%
\typeout{** the default language instead.}%
\else
\language=\csname l@#1\endcsname
\fi
#2}}
\providecommand{\BIBdecl}{\relax}
\BIBdecl

\bibitem{samad_globecom}
S.~Ali, A.~Ferdowsi, W.~Saad, and N.~Rajatheva, ``Sleeping multi-armed bandits
  for fast uplink grant allocation in machine type communications,'' in
  \emph{Proc. IEEE Global Communications Conference (GLOBECOM), Workshop on
  Ultra-High Speed, Low Latency and Massive Connectivity Communication for
  5G/B5G}, Abu Dhabi, UAE, Dec 2018, pp. 1--6.

\bibitem{IoTin5G}
M.~R. Palattella, M.~Dohler, A.~Grieco, G.~Rizzo, J.~Torsner, T.~Engel, and
  L.~Ladid, ``Internet of things in the {5G} era: Enablers, architecture, and
  business models,'' \emph{IEEE Journal on Selected Areas in Communications},
  vol.~34, no.~3, pp. 510--527, March 2016.

\bibitem{mingzheVR}
M.~Chen, W.~Saad, and C.~Yin, ``Virtual reality over wireless networks:
  Quality-of-service model and learning-based resource management,'' \emph{IEEE
  Transactions on Communications}, vol. to appear, 2018.

\bibitem{aidin_its_magazine}
\BIBentryALTinterwordspacing
A.~Ferdowsi, U.~Challita, and W.~Saad, ``Deep learning for reliable mobile edge
  analytics in intelligent transportation systems,'' \emph{CoRR}, vol.
  abs/1712.04135, 2017. [Online]. Available:
  \url{http://arxiv.org/abs/1712.04135}
\BIBentrySTDinterwordspacing

\bibitem{mozaffariUAV}
M.~Mozaffari, W.~Saad, M.~Bennis, and M.~Debbah, ``Unmanned aerial vehicle with
  underlaid device-to-device communications: Performance and tradeoffs,''
  \emph{IEEE Transactions on Wireless Communications}, vol.~15, no.~6, pp.
  3949--3963, June 2016.

\bibitem{dawyM2MMagazine}
Z.~Dawy, W.~Saad, A.~Ghosh, J.~G. Andrews, and E.~Yaacoub, ``Toward massive
  machine type cellular communications,'' \emph{IEEE Wireless Communications},
  vol.~24, no.~1, pp. 120--128, February 2017.

\bibitem{aidin-deeplearning-journal}
A.~Ferdowsi and W.~Saad, ``Deep learning for signal authentication and security
  in massive {I}nternet of {T}hings systems,'' \emph{IEEE Transactions on
  Communications}, vol. to appear, pp. 1--1, 2018.

\bibitem{aidin_ICC}
{A. Ferdowsi} and {W. Saad}, ``Deep learning-based dynamic watermarking for
  secure signal authentication in the {I}nternet of {T}hings,'' in \emph{Proc.
  of 2018 IEEE International Conference on Communications (ICC)}, Kansas City,
  USA, May 2018.

\bibitem{schulz2017latency}
P.~Schulz, M.~Matthe, H.~Klessig, M.~Simsek, G.~Fettweis, J.~Ansari, S.~A.
  Ashraf, B.~Almeroth, J.~Voigt, I.~Riedel, A.~Puschmann, A.~Mitschele-Thiel,
  M.~Muller, T.~Elste, and M.~Windisch, ``Latency critical {IoT} applications
  in {5G}: Perspective on the design of radio interface and network
  architecture,'' \emph{IEEE Communications Magazine}, vol.~55, no.~2, pp.
  70--78, February 2017.

\bibitem{MassiveM2M}
C.~Bockelmann, N.~Pratas, H.~Nikopour, K.~Au, T.~Svensson, C.~Stefanovic,
  P.~Popovski, and A.~Dekorsy, ``Massive machine-type communications in 5{G}:
  physical and {MAC}-layer solutions,'' \emph{IEEE Communications Magazine},
  vol.~54, no.~9, pp. 59--65, September 2016.

\bibitem{surveyofaccess}
M.~T. Islam, A.~e.~M.~Taha, and S.~Akl, ``A survey of access management
  techniques in machine type communications,'' \emph{IEEE Communications
  Magazine}, vol.~52, no.~4, pp. 74--81, April 2014.

\bibitem{RACHM2M2}
A.~Laya, L.~Alonso, and J.~Alonso-Zarate, ``Is the random access channel of
  {LTE} and {LTE-A} suitable for {M2M} communications? a survey of
  alternatives,'' \emph{IEEE Communications Surveys Tutorials}, vol.~16, no.~1,
  pp. 4--16, December 2013.

\bibitem{otpimalACB}
Z.~Wang and V.~W.~S. Wong, ``Optimal access class barring for stationary
  machine type communication devices with timing advance information,''
  \emph{IEEE Transactions on Wireless Communications}, vol.~14, no.~10, pp.
  5374--5387, Oct 2015.

\bibitem{nora}
Y.~Liang, X.~Li, J.~Zhang, and Z.~Ding, ``Non-orthogonal random access for {5G}
  networks,'' \emph{IEEE Transactions on Wireless Communications}, vol.~16,
  no.~7, pp. 4817--4831, July 2017.

\bibitem{RACH_BloomFiltering}
N.~K. Pratas, C.~Stefanovic, G.~C. Madueno, and P.~Popovski, ``Random access
  for machine-type communication based on bloom filtering,'' in \emph{Porc. of
  2016 IEEE Global Communications Conference (GLOBECOM)}, Dec 2016, pp. 1--7.

\bibitem{rach_correlated}
A.~E. Kalor, O.~A. Hanna, and P.~Popovski, ``Random access schemes in wireless
  systems with correlated user activity,'' in \emph{Proc. of IEEE 19th
  International Workshop on Signal Processing Advances in Wireless
  Communications (SPAWC)}, June 2018, pp. 1--5.

\bibitem{newRACH}
N.~Zhang, G.~Kang, J.~Wang, Y.~Guo, and F.~Labeau, ``Resource allocation in a
  new random access for {M2M} communications,'' \emph{IEEE Communications
  Letters}, vol.~19, no.~5, pp. 843--846, May 2015.

\bibitem{madueno2014reliable}
G.~C. Madue{\~{n}}o, {\v{C}}.~Stefanovi{\'c}, and P.~Popovski, ``Reliable
  reporting for massive {M2M} communications with periodic resource pooling,''
  \emph{IEEE Wireless Communications Letters}, vol.~3, no.~4, pp. 429--432, Aug
  2014.

\bibitem{abuzainab2016cognitive}
N.~Abuzainab, W.~Saad, C.~S. Hong, and H.~V. Poor, ``Cognitive hierarchy theory
  for distributed resource allocation in the {I}nternet of {T}hings,''
  \emph{IEEE Transactions on Wireless Communications}, vol.~16, no.~12, pp.
  7687--7702, Dec 2017.

\bibitem{grantfreemassiveNOMA}
R.~Abbas, M.~Shirvanimoghaddam, Y.~Li, and B.~Vucetic, ``Grant-free massive
  {NOMA}: Outage probability and throughput,'' \emph{arXiv preprint
  arXiv:1707.07401}, 2017.

\bibitem{CS_GrantFree_NOMA}
B.~Wang, L.~Dai, Y.~Zhang, T.~Mir, and J.~Li, ``Dynamic compressive
  sensing-based multi-user detection for uplink grant-free noma,'' \emph{IEEE
  Communications Letters}, vol.~20, no.~11, pp. 2320--2323, Nov 2016.

\bibitem{3GPP-fastuplinkgrant}
3GPP, ``Study on latency reduction techniques for {LTE},'' {3rd Generation
  Partnership Project (3GPP)}, Technical Specification (TS) 36.881.

\bibitem{LTE14Outlook}
C.~Hoymann, D.~Astely, M.~Stattin, G.~Wikstrom, J.~F. Cheng, A.~Hoglund,
  M.~Frenne, R.~Blasco, J.~Huschke, and F.~Gunnarsson, ``{LTE} release 14
  outlook,'' \emph{IEEE Communications Magazine}, vol.~54, no.~6, pp. 44--49,
  June 2016.

\bibitem{Samad-fastuplinkgrant}
S.~Ali, N.~Rajatheva, and W.~Saad, ``Fast uplink grant for machine type
  communications: Challenges and opportunities,'' \emph{arXiv preprint
  arXiv:1801.04953}, 2018.

\bibitem{MTCsourceTraffic}
M.~Laner, P.~Svoboda, N.~Nikaein, and M.~Rupp, ``Traffic models for machine
  type communications,'' in \emph{Proc. of the Tenth International Symposium on
  Wireless Communication Systems}, Ilmenau, Germany, Aug 2013, pp. 1--5.

\bibitem{m2mtrafficstudy}
M.~Centenaro and L.~Vangelista, ``A study on {M2M} traffic and its impact on
  cellular networks,'' in \emph{2015 IEEE 2nd World Forum on Internet of Things
  (WF-IoT)}, Dec 2015, pp. 154--159.

\bibitem{periodicpattern}
J.~Adhikari and P.~Rao, ``Identifying calendar-based periodic patterns,''
  \emph{Emerging paradigms in machine learning}, pp. 329--357, 2013.

\bibitem{DI_letter}
S.~Ali, W.~Saad, and N.~Rajatheva, ``A directed information learning framework
  for event-driven {M2M} traffic prediction,'' \emph{IEEE Communications
  Letters}, pp. 1--1, 2018.

\bibitem{brownPredictive}
J.~Brown and J.~Y. Khan, ``A predictive resource allocation algorithm in the
  {LTE} uplink for event based {M2M} applications,'' \emph{IEEE Transactions on
  Mobile Computing}, vol.~14, no.~12, pp. 2433--2446, Dec 2015.

\bibitem{sutton1998reinforcement}
R.~S. Sutton, A.~G. Barto \emph{et~al.}, \emph{Reinforcement learning: An
  introduction}.\hskip 1em plus 0.5em minus 0.4em\relax MIT press, 1998.

\bibitem{MABsinSmallCells}
S.~Maghsudi and E.~Hossain, ``Multi-armed bandits with application to {5G}
  small cells,'' \emph{IEEE Wireless Communications}, vol.~23, no.~3, pp.
  64--73, June 2016.

\bibitem{MABd2d}
S.~Maghsudi and S.~Stańczak, ``Channel selection for network-assisted {D2D}
  communication via no-regret bandit learning with calibrated forecasting,''
  \emph{IEEE Transactions on Wireless Communications}, vol.~14, no.~3, pp.
  1309--1322, March 2015.

\bibitem{MABdist}
S.~Maghsudi and E.~Hossain, ``Distributed user association in energy harvesting
  small cell networks: A probabilistic bandit model,'' \emph{IEEE Transactions
  on Wireless Communications}, vol.~16, no.~3, pp. 1549--1563, March 2017.

\bibitem{MABmutliusers}
Y.~Gai, B.~Krishnamachari, and R.~Jain, ``Learning multiuser channel
  allocations in cognitive radio networks: A combinatorial multi-armed bandit
  formulation,'' in \emph{Proc. of IEEE Symposium on New Frontiers in Dynamic
  Spectrum (DySPAN)}, Singapore, Singapore, April 2010, pp. 1--9.

\bibitem{ValueofInformation}
\BIBentryALTinterwordspacing
C.~Bisdikian, L.~M. Kaplan, and M.~B. Srivastava, ``On the quality and value of
  information in sensor networks,'' \emph{ACM Trans. Sen. Netw.}, vol.~9,
  no.~4, pp. 48:1--48:26, Jul. 2013. [Online]. Available:
  \url{http://doi.acm.org/10.1145/2489253.2489265}
\BIBentrySTDinterwordspacing

\bibitem{mimoMatlab}
Y.~S. Cho, J.~Kim, W.~Y. Yang, and C.~G. Kang, \emph{MIMO-OFDM wireless
  communications with MATLAB}.\hskip 1em plus 0.5em minus 0.4em\relax John
  Wiley \& Sons, 2010.

\bibitem{3GPPPathLossModel}
3GPP, ``Radio frequency ({RF}) requirements for {LTE} pico node {B},'' {3rd
  Generation Partnership Project (3GPP)}, Technical Specification (TS) 36.931.

\bibitem{Gompertz}
D.~Juki{\'c}, G.~Kralik, and R.~Scitovski, ``Least-squares fitting {G}ompertz
  curve,'' \emph{Journal of Computational and Applied Mathematics}, vol. 169,
  no.~2, pp. 359--375, 2004.

\bibitem{kleinberg2010regret}
R.~Kleinberg, A.~Niculescu-Mizil, and Y.~Sharma, ``Regret bounds for sleeping
  experts and bandits,'' \emph{Machine learning}, vol.~80, no. 2-3, pp.
  245--272, 2010.

\bibitem{MingzheSurvey}
\BIBentryALTinterwordspacing
M.~Chen, U.~Challita, W.~Saad, C.~Yin, and M.~Debbah, ``Machine learning for
  wireless networks with artificial intelligence: {A} tutorial on neural
  networks,'' \emph{CoRR}, vol. abs/1710.02913, 2017. [Online]. Available:
  \url{http://arxiv.org/abs/1710.02913}
\BIBentrySTDinterwordspacing

\bibitem{auer2002finite}
P.~Auer, N.~Cesa-Bianchi, and P.~Fischer, ``Finite-time analysis of the
  multiarmed bandit problem,'' \emph{Machine learning}, vol.~47, no. 2-3, pp.
  235--256, 2002.

\end{thebibliography}
\end{document}